\documentclass[11pt, a4paper]{article}
\usepackage[margin=3.5cm]{geometry}
\usepackage[table]{xcolor}
\usepackage[ruled, english, linesnumbered]{algorithm2e}
\usepackage[inline]{enumitem}
\usepackage[T1]{fontenc}
\usepackage[utf8]{inputenc}
\usepackage{adjustbox}
\usepackage{arydshln}
\usepackage{array}
\usepackage{caption}
\usepackage{comment}
\usepackage[hidelinks]{hyperref}
\usepackage{mathrsfs}
\usepackage{mathtools}
\usepackage{amsmath,amssymb,amsthm}
\usepackage{musicography}
\usepackage{stmaryrd}
\usepackage{tikz, tikz-cd}
\usetikzlibrary{patterns}
\usepackage{array,multirow}
\usepackage{verbatim}
\usepackage{xspace}
\usepackage{microtype}
\usepackage{threeparttable}
\usepackage[numbers]{natbib}
\usepackage{ebgaramond,ebgaramond-maths}

\newcommand{\B}{\mathcal{B}}
\newcommand{\id}{\text{\rm id}}

\let\Im\undefined

\DeclareMathOperator{\Hom}{Hom}
\DeclareMathOperator{\Im}{im}

\DeclareMathOperator{\Ker}{ker}
\DeclareMathOperator{\Coker}{coker}

\DeclareMathOperator{\Frac}{Frac}
\DeclareMathOperator{\Gal}{Gal}

\DeclareMathOperator{\rgcd}{rgcd}

\newcommand{\Fp}{\mathbb{F}_p}

\newcommand{\Fq}{{\mathbb{F}_q}}

\newcommand{\Kbar}{{\overline{K}}}
\newcommand{\Ksep}{K^{\text{\rm sep}}}
\newcommand{\NN}{\mathbb{Z}_{\geqslant 0}}

\newcommand{\ZZ}{\mathbb{Z}}
\newcommand{\FF}{\mathbb{F}}

\renewcommand{\geq}{\geqslant}
\renewcommand{\leq}{\leqslant}

\renewcommand{\a}{\mathfrak{a}}
\newcommand{\m}{\mathfrak{m}}
\newcommand{\A}{\mathcal{A}}

\newcommand{\Aq}{\mathcal{\A_{\q}}}

\newcommand{\p}{\mathfrak{p}}
\newcommand{\q}{\mathfrak{q}}
\newcommand{\E}{\mathbb{E}}
\newcommand{\T}{\mathbb{T}}
\newcommand{\M}{\mathbb{M}}
\DeclareMathOperator{\Frob}{Frob}
\newcommand{\Ktau}{K\{\tau\}}

\newcommand{\detideal}{\mathfrak{det}\,}
\newcommand{\norm}{\mathfrak n}

\newcommand{\Nrd}{N_{\text{\rm rd}}}
\newcommand{\lc}{\text{\rm lc}}

\newcommand{\Otilde}{\ensuremath{\mathop{O\hspace{0.2ex}\tilde{~}}}}
\newcommand{\Opower}{\ensuremath{\mathop{O^\bullet}}}
\newcommand{\FMFF}{\texttt{F}-\texttt{MFF}\xspace}
\newcommand{\FMKU}{\texttt{F}-\texttt{MKU}\xspace}
\newcommand{\FCSA}{\texttt{F}-\texttt{CSA}\xspace}

\newcommand{\MotiveCoordinates}{\textsc{MotiveCoordinates}{}}
\newcommand{\MotiveTauAction}{\textsc{MotiveTauAction}{}}
\newcommand{\MotiveMatrix}{\textsc{MotiveMatrix}{}}

\newcommand{\CSAMatrix}{\textsc{Matrix-CSA}{}}
\newcommand{\FrobeniusCharpoly}{\textsc{FrobeniusCharpoly-CSA}{}}
\newcommand{\IsogenyNorm}{\textsc{IsogenyNorm}{}}
\newcommand{\EndomorphismCharpoly}{\textsc{EndomorphismCharpoly}{}}

\DeclareMathOperator{\SM}{SM}
\newcommand{\SMgeq}{{\SM^{\geqslant 1}}}

\newcommand{\coord}{\Gamma}

\AtBeginEnvironment{algorithm}{
  \DontPrintSemicolon
  \SetKwInput{KwIn}{\emph{Input}}
  \SetKwInput{KwOut}{\emph{Output}}
  \SetKwSty{normalfont}
  \SetKwFor{For}{For}{}{}
  \SetKw{KwRet}{Return}
  \SetKwIF{If}{ElseIf}{Else}{If}{then}{Else if}{Else}{Endif}
}

\renewcommand{\epsilon}{\varepsilon}

\definecolor{purple}{rgb}{0.6,0,0.6}

\newcounter{thintro}

\newtheorem{theo*}[thintro]{Theorem}
\newtheorem{theo}{Theorem}[section]

\newtheorem{lem}[theo]{Lemma}
\newtheorem{prop}[theo]{Proposition}
\newtheorem{cor}[theo]{Corollary}

\theoremstyle{definition}
\newtheorem{rem}[theo]{Remark}
\newtheorem{ex}[theo]{Example}
\newtheorem{deftn}[theo]{Definition}

\newcommand{\functiondef}[4]{
  \begin{center}
    \begin{tabular}{rcl}
      $#1$ & $\to$     & $#2$ \\
      $#3$ & $\mapsto$ & $#4$
    \end{tabular}
  \end{center}
}

\newcommand{\functiondefname}[5]{
  \begin{center}
    \begin{tabular}{rrcl}
      $#1$ : & $#2$ & $\to$     & $#3$ \\
             &   $#4$ & $\mapsto$ & $#5$
    \end{tabular}
  \end{center}
}

\newcommand{\functionname}[3]{
  \[
      #1 : \quad #2 \; \to \; #3
  \]
}

\newcommand{\ttau}{t}

\usetikzlibrary{trees, positioning}

\begin{document}

  \title{Algorithms for computing norms and characteristic polynomials on general Drinfeld~modules}
  \date\today

  \author{%
    Xavier Caruso\footnote{Université de Bordeaux, CNRS, INRIA, 351, cours de la Libération, 33405 Talence, France},
    Antoine Leudière\footnote{Université de Lorraine, INRIA, CNRS, 615 rue du Jardin Botanique, 54600 Villers-lès-Nancy, France}%
  }

  \maketitle

  \begin{abstract}
    We provide two families of algorithms to compute characteristic 
    polynomials of endomorphisms and norms of isogenies of Drinfeld 
    modules. Our algorithms work for Drinfeld modules of any rank,
    defined over any base curve.
    When the base curve is $\mathbb P^1_{\Fq}$, we do a thorough
    study of the complexity, demonstrating that our algorithms are,
    in many cases, the most asymptotically performant.
    The first family of algorithms relies on the correspondence
    between Drinfeld modules and Anderson motives, reducing the computation to
    linear algebra over a polynomial ring. The second family, available only for the Frobenius endomorphism,
    is based on a formula expressing
    the characteristic polynomial of the Frobenius as a
    reduced norm in a central simple algebra.
  \end{abstract}

\setcounter{tocdepth}{2}
\tableofcontents

\section*{Introduction}

Drinfeld modules were introduced in 1974 to serve as the foundations of the
class field theory of function fields~\cite{drinfeld-paper}. Although they were
initially considered as mathematical abstract objects, recent papers
highlighted a growing interest for the computational aspects in these topics:
in the recent years, a PhD thesis~\cite{caranay-thesis} and at least three
papers focused on the algorithmics of Drinfeld modules~\cite{brenner_computing_2020, musleh-schost-1,
musleh-schost-2}.
Due to their striking similarities with elliptic curves, Drinfeld modules
were considered several times for their applications in cryptography~\cite{joux_drinfeld_2019,
scanlon_public_2001, bombar_codes_2022, leudiere_computing_2024}. Other
applications saw them being used to efficiently factor polynomials in $\Fq[T]$
\cite{doliskani-narayanan-schost}.

The present paper is a contribution to the algorithmic toolbox of 
Drinfeld modules. More precisely, we focus on the effective and 
efficient computation of characteristic polynomials of endomorphisms
of Drinfeld modules, as well as norms of general isogenies.

\paragraph{Context.}

Before going deeper into our results, we recall briefly the purpose and the
most significant achievements of the theory of Drinfeld modules. Classical
class field theory aims at describing abelian extensions of local and global
fields, using information available solely at the field's level
\cite{chevalley_theorie_1940, conrad_history_2009}. Premises of the theory go
back to Gau\ss' \emph{Disquisitiones Arithmeticae}, and in 1853, Kronecker
stated the famous Kronecker-Weber theorem: every abelian number field lies
inside a cyclotomic field \cite{kronecker-weber, hilbert_neuer_1932}. Another
crucial theorem from class field theory is the Kronecker \emph{Jugendtraum},
relating maximal abelian unramified extensions of quadratic imaginary number
fields and the theory of complex multiplication of elliptic curves. More
generally, a result conjectured by Hilbert, and proved by Takagi in 1920
\cite{takagi_collected_2014}, asserts that every number field $K$ is contained
inside a maximal abelian unramified extension $H$ whose class group is isomorphic to
$\Gal(H/K)$. The field $H$ is called the \emph{Hilbert class field} of $K$ and,
apart from abelian number fields and imaginary quadratic number fields, it is
generally hard to describe, yet even to compute.

One of the goals of the introduction of Drinfeld modules is to set up an
analogue of these results for function fields. They were also instrumental in
proving a special case of the Langlands program for $\mathrm{GL}_r$ of a function
field (see \cite{drinfeld-paper} for $r=2$ and \cite{laumon_cohomology_1995}
for general $r$). Lafforgue proved the global Langlands correspondence for
$\mathrm{GL}_r$ of a function field using generalizations of Drinfeld modules
called \emph{shtukas}. He was awarded the Fields medal for this
work~\cite{lafforgue}.

A Drinfeld module is an object defined in the following
setting: a base curve $C$ over $\Fq$ which is projective, smooth and
geometrically connected (\emph{e.g.} $C = \mathbb P^1_{\Fq}$); a fixed point
$\infty$ of~$C$; the ring $A$ of rational functions on $C$ regular outside
$\infty$ (\emph{e.g.} $A = \Fq[T]$); a base field $K$ with a structure of
$A$-algebra given by an $\Fq$-algebra morphism $\gamma: A \to K$. We then talk
about Drinfeld $A$-modules.
In this setting, we define $\p = \ker \gamma$; it is an ideal of $A$ acting as a function
field analogue of the more classical characteristic $p$.
An important feature of Drinfeld modules is that
they endow the algebraic closure $\Kbar$ of $K$ with a structure of $A$-module. When $A =
\Fq[T]$, this structure surprisingly resembles to the $\ZZ$-module
structure on the points of an elliptic curve. Important references on Drinfeld
modules include \cite{gek91, gos98, rosen_number_2002,
poonen_introduction_2022, villa_salvador_topics_2006, hayes_brief_2011,
papikian_drinfeld_2023}.

The simplest Drinfeld modules are the rank $1$ Drinfeld modules over the curve
$\mathbb P^1_\Fq$, where $K$ is the function field $\Fq(T)$, \emph{i.e.} the
Drinfeld $\Fq[T]$-modules of rank $1$ over $\Fq(T)$. They were studied by
Carlitz \cite{carlitz_certain_1935}, and provide function field analogues of
roots of unity, and consequently, of cyclotomic fields; the analogue of the
Kronecker-Weber theorem was subsequently proved by Hayes
\cite{hayes_explicit_1974}. Coming to the \emph{Jugendtraum}, we need to go to
Drinfeld modules of rank $1$ over general curves and Drinfeld $\Fq[T]$-modules
or rank $2$ over finite fields. The latter have a theory of complex
multiplication which shares many similarities with that of elliptic curves over
finite fields. As an illustration, we mention that the endomorphism ring of
such a Drinfeld module is either an order in a quadratic imaginary function field
or a maximal order in a quaternion algebra.

\paragraph{Algorithmic results.}

Like in the classical setting, the theory of complex multiplication of Drinfeld
modules depends heavily on the notion of \emph{characteristic polynomial of the
Frobenius endomorphism}, which we compute in this paper. This polynomial lies in 
$A[X]$ and is an invariant of primary importance: it determines the
isogeny class of the underlying Drinfeld module, it controls the theory
of complex multiplication and it is the main building block in the
construction of the attached $L$-series (see~\cite{caruso-gazda} and 
references therein). Moreover, in the case of rank $2$ Drinfeld modules over
$\Fq[T]$, it determines if a Drinfeld module is ordinary or supersingular,
as happens with elliptic curves. The characteristic polynomial of the Frobenius also defines curves and
extensions that naturally arise in the class field theory of function fields
\cite{leudiere_computing_2024}. More generally, characteristic polynomials can
be defined for any endomorphism in any rank and over any base.

In the present paper, we design algorithms for computing the characteristic
polynomial of any endomorphism of a Drinfeld module on the one hand,
and for computing the norm of any isogeny between Drinfeld modules on the other
hand. When $A = \Fq[T]$, we moreover do a thorough analysis of their 
complexity.
To state our complexity results, it is convenient to use Laudau's
$O$-notation and some of its variants. Precisely, if $f$ and $g$
are two positive quantities depending on parameters, we write
\begin{itemize}[itemsep=0.5ex,parsep=0ex,topsep=0.5ex]
\item $g \in O(f)$ if there exists an absolute positive constant $C$ 
such that $g \leq C {\cdot} f$,
\item $g \in \Otilde(f)$ if there exist absolute positive constant $C$ and
$k$ such that $g \leq C {\cdot} f \log^k f$,
\item $g \in \Opower(f)$ if, for all $\varepsilon > 0$, there exists 
a positive constant $C_\varepsilon$ 
such that $g \leq C_\varepsilon {\cdot} f^{1 + \varepsilon}$,
\end{itemize}
where all inequalities are required to hold true for \emph{all} choices 
of parameters.

Let also $\omega \in [2,3]$ denote a feasible exponent for matrix 
multiplication; by this, we mean that we are given an algorithm which is 
able to compute the product of two $n \times n$ matrices over a ring $R$ 
for a cost of $O(n^\omega)$ operations in~$R$. The naive algorithm leads
to $\omega = 3$; however, better algorithms do exist and, currently, the one
with the lowest $\omega$ does so for $\omega$ approximately equal to
$2.37188$~\cite{matrices:omega}.
Similarly, let $\Omega$ be a feasible exponent for the computation of the
characteristic polynomial of a matrix over polynomials rings over a field.
Using Kaltofen and Villard's algorithm, it is known that one can reach
$\Omega < 2.69497$~\cite{matrices:kaltofen-villard}.
If $K$ is a finite extension of $\Fq$ of degree $d$, we also denote by
$\SMgeq(n,d)$ a log-concave function with respect to the variable $n$
having the following property:
the number of operations in $\Fq$\footnote{Here,
we assume that applying the Frobenius of $K$ counts for $\Otilde(d)$
operations
in $\Fq$, see \S \ref{subsec:ore-computations} for more details.} needed
for multiplying two Ore polynomials in $K\{\tau\}$ of degree~$n$ is
in $\Otilde(\SMgeq(n,d))$.

Our first result is about the computation of the characteristic polynomial of
an endomorphism of a Drinfeld module.

\begin{theo*}[see Theorems~\ref{theo:motive-charpoly-comp} and~\ref{theo:charpoly-finitefield}]
\label{thintro:charpoly}
  Let $\phi$ be a Drinfeld $\Fq[T]$-module of rank $r$ over a field $K$, and
  let $u$ be an endomorphism\footnote{We refer to \S \ref{ssec:backgrounddrinfeld} 
  for the definition of an endomorphism of a Drinfeld module and of its degree.}
  of $\phi$ of degree $n$.
  The characteristic polynomial of $u$ can be computed for a cost of
  $\Otilde(n^2 + (n+r)r^{\Omega - 1})$
  operations in $K$ and $O(n^2 + r^2)$
  applications of the Frobenius.

  Moreover, when $K$ is a finite extension of $\Fq$ of degree~$d$, 
  the characteristic polynomial of $u$ can be computed for a cost of
  \[
    \Otilde(d \log^2 q) + 
    \Opower\big( \big(\SMgeq(n, d) + ndr + (n + d)r^\omega\big)\cdot\log q\big)
  \]
  bit operations.
\end{theo*}

We then study more particularly the special case of the Frobenius
endomorphism (which is only defined when $K$ is a finite field),
for which we provide three different algorithms that we call \FMFF,
\FMKU and \FCSA respectively.

\begin{theo*}
\label{thintro:frob}
  Let $\phi$ be a Drinfeld $\Fq[T]$-module of rank $r$ over a finite 
  extension $K$ of $\Fq$ of degree~$d$. The characteristic polynomial of the Frobenius
  endomorphism of $\phi$ can be computed for a cost of either
  \begin{itemize}[itemsep=0.5ex,parsep=0ex,topsep=0.5ex]
    \item {\normalfont [\FMFF algorithm, see \S \ref{sssec:charpoly-finitefield}]\hspace{1ex}}
      $\Otilde(d \log^2 q) + \Opower\big((\SMgeq(d, d) + d^2r + dr^\omega)\cdot \log q\big)$, or

    \item {\normalfont [\FMKU algorithm, see \S \ref{sssec:charpoly-kedlaya-umans}]\hspace{1ex}}
      $\Otilde(d \log^2 q) + \Opower\big((d^2r^{\omega - 1} + dr^\omega)\cdot \log q\big)$, or

    \item {\normalfont [\FCSA algorithm, see \S \ref{ssec:CSAP1}]\hspace{1ex}}
      $\Otilde(d \log^2 q) + \Opower(r d^\omega \log q)$
  \end{itemize}
  bit operations.
\end{theo*}

We finally come to general isogenies between different Drinfeld modules.
In this case, the characteristic polynomial is not well-defined, but the
norm is.

\begin{theo*}[see Theorems~\ref{theo:iso-norm} and~\ref{theo:iso-norm-finite}]
  Let $\phi$ and $\psi$ be two Drinfeld $\Fq[T]$-modules of rank $r$ over a
  field $K$, and let $u : \phi \to \psi$ be an isogeny of degree $n$. 
  The norm of $u$ can be computed for a cost of 
  $\Otilde(n^2 + nr^{\omega-1} + r^\omega)$
  operations in $K$ and $O(n^2 + r^2)$ applications of the Frobenius.

  Moreover, when $K$ is a finite extension of $\Fq$ of degree~$d$, 
  the norm of $u$ can be computed for a cost of 
  \[
    \Otilde(d \log^2 q) + 
    \Opower\big(\big(\SMgeq(n, d) + ndr + n\min(d,r)r^{\omega-1} + dr^\omega\big)
    \cdot \log q \big)
  \]
  bit operations.
\end{theo*}

Moreover, we propose extensions of all our algorithms to Drinfeld 
modules defined over a general curve $C$ (and not just $\mathbb P^1_{\Fq}$). 
However, we do not carry out, in the present paper, a thorough study of 
the complexity in this general setting.

Finally, we mention that, in the case of $\mathbb P^1_{\Fq}$, our algorithms
have been implemented in SageMath~\cite{software-presentation} and will be 
hopefully publicly available soon in the standard distribution.
Meanwhile, the interested user may read tutorials and try out our 
software package online on the platform plm-binder~at:

\medskip

\noindent
\hfill\url{https://xavier.caruso.ovh/notebook/drinfeld-modules}\hfill\null

\medskip

\paragraph{Comparison with previous results.}

To the authors' knowledge, it is the first time that algorithms are 
presented for Drinfeld modules defined over a general curve; so far,
only the case of $\mathbb P^1_{\Fq}$ was adressed.
Also, we are not aware of previous works on the explicit computations
of norms of general isogenies between different Drinfeld modules.

In contrast, the question of the explicit computation of the 
characteristic polynomial of the Frobenius endomorphism, especially 
in the case of rank $2$, was already considered by many
authors~\cite{narayanan, doliskani-narayanan-schost, garai-papikian, musleh-schost-1, musleh-schost-2}.
Our algorithms for this task are however new and they turn out to be
competitive for a large range of parameters.
More precisely, prior to our work, the most efficient algorithm was due 
to Musleh and Schost~\cite{musleh-schost-2}.
Depending on the relative values of $r$, $d = [K: \Fq]$ and $m = 
\deg(\p)$, all four algorithms (\FMFF, \FMKU, \FCSA and Musleh-Schost's
algorithm) achieve the best asymptotic complexity in at least one regime, 
as shown in Figure~\ref{fig:comparison}.
As a rule of thumb, the reader can memorize that our algorithms are
better when $r \gg \sqrt d$ (or even $r \gg 
d^{0.431}$ if one takes into account fast algorithms for matrix 
multiplication); on the contrary, when $r \ll \sqrt d$, our algorithms 
may still be competitive, depending on the relative values of $\log(m) / 
\log(d)$ and $\log(r)/\log(d)$.

For a more complete review on existing algorithms and comparison 
between complexities, we refer to the tables of Appendix~\ref{appendix:review}
(page~\pageref{appendix:review}).

\def\cA{(0.3805, 0.761)}
\def\cB{(0.3805, 0)}
\def\cC{(0.4308, 1)}
\def\ph{\vphantom{$A^A_A$}}

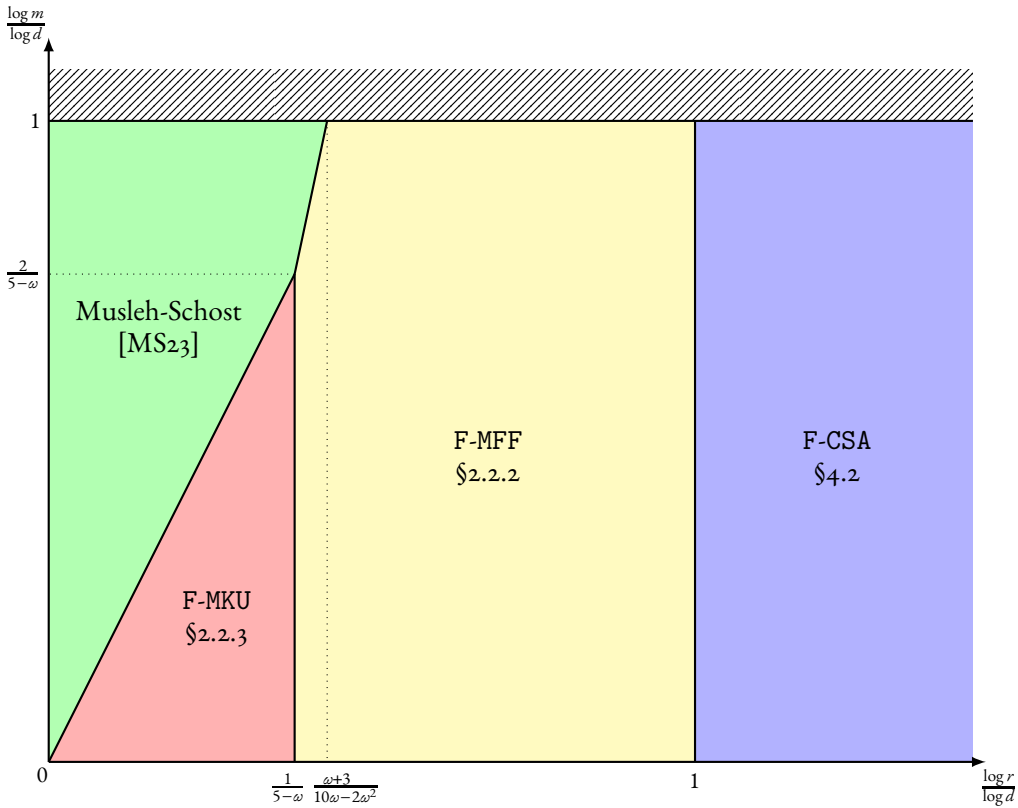
\begin{figure}
\hfill%
\begin{tikzpicture}[scale=7.7]
\fill[pattern=north east lines] (0,1) rectangle (1.43,1.08);
\fill[green!30] (0,0)--\cA--\cC--(0,1);
\fill[red!30] (0,0)--\cA--\cB;
\fill[yellow!30] \cB--\cA--\cC--(1,1)--(1,0);
\fill[blue!30] (1,0) rectangle (1.43,1);
\begin{scope}[thick]
\draw[-latex] (0,0)--(1.45,0);
\draw[-latex] (0,0)--(0,1.13);
\draw (0,1)--(1.43,1);
\draw (0,0)--\cA--\cC;
\draw \cA--\cB;
\draw (1,0)--(1,1);
\end{scope}
\node[scale=0.8] at (1.47,-0.04) { $\frac{\log r}{\log d}$ };
\node[scale=0.8] at (-0.035, 1.15) { $\frac{\log m}{\log d}$ };
\draw[dotted] \cA--(0, 0.761);
\draw[dotted] \cC--(0.4308, 0);
\node[scale=0.8] at (-0.01,-0.02) { $0$ };
\node[scale=0.8] at (-0.02,1) { $1$ };
\node[scale=0.8] at (1,-0.03) { $1$ };
\node[scale=0.8] at (0.37,-0.04) { $\frac 1{5-\omega}$ };
\node[scale=0.8] at (-0.04, 0.761) { $\frac 2{5-\omega}$ };
\node[scale=0.8] at (0.46,-0.04) { $\frac {\omega+3\hspace{3mm}}{10\omega - 2 \omega^2}$ };
\node at (0.17,0.7) { \ph Musleh-Schost };
\node at (0.17,0.65) { \ph \cite{musleh-schost-2} };
\node at (0.26,0.25) { \ph \FMKU };
\node at (0.26,0.2) { \ph \S \ref{sssec:charpoly-kedlaya-umans} };
\node at (0.68,0.5) { \ph \FMFF };
\node at (0.68,0.45) { \ph \S \ref{sssec:charpoly-finitefield} };
\node at (1.22,0.5) { \ph \FCSA };
\node at (1.22,0.45) { \ph \S \ref{ssec:CSAP1} };
\end{tikzpicture}%
\hfill\null

\caption{The best algorithm for computing the characteristic polynomial 
of the Frobenius endomorphism, depending on the size of $r$, $d$ and $m$.\\
\emph{Assumptions:} $2 \leq \omega \leq 3$ and $\omega \leq \Omega \leq \omega + 1$.}
\label{fig:comparison}
\end{figure}

\paragraph{Anderson motives.}

The main theoretical input upon which all our algorithms are based is the
\emph{motive} attached to a Drinfeld module, introduced by Anderson in
1986~\cite{anderson_t-motives_1986} (see also~\cite{gos98, weil-pairing,
grishkov_introduction_2020}). In the classical setting of algebraic geometry,
Grothendieck describes the motive $\M(X)$ of an algebraic variety $X$ as the
ultimate object able to encode all the ``linear'' properties of~$X$. Since
characteristic polynomials and norms are obviously constructions of linear
nature, we expect to be able to recover them at the level of motives. However,
in the classical setting, motives are usually quite complicated objects, often
defined by accumulating subtle categorical constructions. More or less, this
totally prevents using them for algorithmic applications.

It is striking that the situation for Drinfeld modules is much more tractable:
the Anderson motive $\M(\phi)$ of a Drinfeld module $\phi$ is a very 
explicit object---concretely, it is just $K\{\tau\}$ equipped with extra 
structures---which is very well-adapted to algorithmic manipulations. 
However, $\M(\phi)$~exhibits all the theoretical features one expects; in 
particular, it retains all the information we need on characteristic
polynomials of endomorphisms and norms of isogenies.
In the present paper, we make an intersive use of this yoga. In 
particular, we highlight that our methods are not an adaptation of 
existing methods from elliptic curves.

More precisely, an endomorphism $u$ of a Drinfeld module corresponds to 
a linear endomorphism $\M(u)$ at the level of Anderson motives. It is 
moreover a well-known fact that the characteristic polynomial of $\M(u)$ 
agrees with that of $u$ (see \cite[Proposition~3.6.7]{papikian_drinfeld_2023}
for the case of $\mathbb P^1_{\Fq}$). In the present paper, we give a new proof 
of this theorem, and extend it to general isogenies, establishing that 
the norm of an isogeny $u$ is the ideal generated by the determinant
of~$\M(u)$ (see Theorem~\ref{th:norm-general}).
We then use this result to reduce the computations we are interested in 
to the computation of the determinant or the characteristic polynomial
of an actual matrix. In the case of $\mathbb P^1_{\Fq}$, this is immediate 
since Anderson motives are free over $K[T]$, with an explicit 
canonical basis.
For a general curve, Anderson motives are not always free but only
projective, which induces technical difficulties for algorithmics. 
Although it should be doable to tackle these issues head-on, we 
choose to work around them by reducing the problem to the case of 
$\mathbb P^1_{\Fq}$ treated previously.

\paragraph{The \emph{central simple algebra} method.}

The previous discussion applies to all our algorithms, except the
algorithm \FCSA which is different in nature: it is based on a
formula interpreting the characteristic polynomial of the Frobenius
endomorphism as a reduced norm in some well-suited central simple
algebra (see Remark~\ref{rem:csa}).
This reduces the computation of the characteristic
polynomial of the Frobenius endomorphism to the computation of a
reduced characteristic polymonial which, using classical techniques,
further reduces to the computation of the characteristic polynomial
of an actual matrix over of size $d \times d$ (with $d = [K:\Fq]$ as
above) over~$\Fq[T]$.

\medskip

To conclude, we would like to mention that, on the theoretical side, 
Anderson motives are not only a powerful tool for studying Drinfeld 
modules; they are nowadays considered as a vast generalization of 
Drinfeld modules, providing more flexibility in the constructions
and having their own interest. The methods presented in this article strongly 
suggest that designing algorithms in the framework of general Anderson 
motives is completely in our reach (and maybe easier!). We then do
believe that time is ripe to go beyond Drinfeld modules and start 
working with Anderson motives at the algorithmic level.

\paragraph{Acknowledgements.}

We thank Pierre-Jean Spaenlehauer and Emmanuel Thomé for their guidance.
We thank Cécile Armana, Alain Couvreur, Quentin Gazda, Federico Pellarin
and Floric Tavarès-Ribeiro for helpful discussions. We thank Mihran Papikian
for his comments.
This work benefited from the financial support of the ANR projects
CLap-CLap (ANR-18-CE40-0026-01), Barracuda (ANR-21-CE39-0009) and
PadLEfAn (ANR-22-CE40-0013), as well as of the France 2030 program managed by the French National Research Agency under grant agreement No. ANR-22-PETQ-0008 PQ-TLS.

\section{Background}

This section serves as a gentle preliminary part in which we 
introduce the setup of this article. On the theoretical side, we 
recall basic definitions and constructions on Drinfeld modules while,
on the computational side, we specify our complexity model and discuss
several algorithmic primitives we shall constantly use throughout 
this article.

\subsection{Drinfeld modules}
\label{ssec:backgrounddrinfeld}

Throughout this paper, we fix a finite field $\Fq$ of cardinality $q$. Let $C$
be smooth, projective, geometrically connected curve over $\Fq$. Let $\infty$
be a distinguished closed point on $C$ and let $A$ denote the ring of rational
functions on $X$ that are regular outside $\infty$. If $F$ is an extension of
$\Fq$, we write $A_F = F \otimes_{\Fq} A$. Thanks to our assumptions on $C$,
the ring $A_F$ is a Dedekind domain. We recall that the \emph{degree} of an
ideal $\a$ of $A_F$, denoted by $\deg(\a)$, is defined as the $F$-dimension
of $A_F/\a$.
For $a \in A_F$, we will often write $\deg(a)$ for $\deg(a A_F)$.

We consider an extension $K$ of $\Fq$ and fix an algebraic
closure $\Kbar$ of $K$. We fix in addition a homomorphism of 
$\Fq$-algebras
\[
  \gamma: A \to K.
\]
The kernel of $\gamma$, a prime ideal of $A$, is denoted by $\p$ and referred to as the
\emph{characteristic}.
An ideal of $A$ is said \emph{away from the characteristic} if it is
coprime to $\p$. Finally, we let $\Ktau$ be
the algebra of \emph{Ore polynomials} over $K$ in $\tau$, in which
the multiplication is twisted according to the rule $\tau a = a^q
\tau$ for all $a \in K$.

\subsubsection{Drinfeld modules and isogenies}

We define Drinfeld modules and their morphisms.

\begin{deftn}[Drinfeld modules]

  A \emph{Drinfeld $A$-module} (or a \emph{Drinfeld module} for short)
  over $K$ is a ring homomorphism
  \[
    \phi: A \to \Ktau
  \]
  whose constant coefficient agrees with $\gamma$ and whose image is not
  contained in $K$.

\end{deftn}

For $a \in A$, we write $\phi_a$ for $\phi(a)$.
By definition, the \emph{rank} of $\phi$ is the unique positive integer 
$r$ such that $\deg(\phi_a) = r \deg(a)$ for all $a \in A$ (see 
\cite[Definition~1.1]{gek91}).

\begin{ex}

  The simplest Drinfeld modules are those for which $C = \mathbb P^1_{{\Fq}}$ and
  $\infty$ is the point at infinity, \emph{i.e.} $A = \Fq[T]$.
  In this case, a Drinfeld module $\phi$ of rank $r$ is defined by 
  the datum of an Ore polynomial
  \[
    \phi_T = \gamma(T) + g_1 \tau \cdots + g_r \tau^r,
  \]
  with $g_1, \ldots,
  g_r \in K$ and $g_r \neq 0$. The \emph{Carlitz module} is commonly defined as
  the rank one Drinfeld $\Fq[T]$-module defined by $T + \tau$ on $K = \Fq(T)$,
  or by $T - \tau$, which was originally studied.

\end{ex}

\begin{deftn}[Morphisms]
Let $\phi, \psi$ be two Drinfeld modules.
  A \emph{morphism} 
$u: \phi \to \psi$ is, by definition, an Ore polynomial $u$ such that
$u \phi_a = \psi_a u$
for every $a \in A$. An \emph{isogeny} is a nonzero morphism.
\end{deftn}

This definition equips the class of Drinfeld
modules with a structure of category, in which the composition is given by
the product in the ring of Ore polynomials.
We say that $\phi$ and $\psi$ are \emph{isogenous} if there exists
an isogeny between $\phi$ and $\psi$. One checks that two isogenous Drinfeld
modules have the same rank.
For any $a \in A$, $\phi_a$ defines an endomorphism of $\phi$.
If $K$ is a finite field of degree $d$ over $\Fq$, then $\tau^d$ defines an 
endomorphism called the \emph{Frobenius endomorphism} of $\phi$; 
it is denoted by $F_\phi$.

Let $u : \phi \to \psi$ be an isogeny defined by the degree $n$ Ore polynomial
\[
  u = u_0 + u_1 \tau + \cdots + u_n \tau^n.
\]
We say that $n$ is the $\tau$-degree of $u$.
By definition, the \emph{height} of $u$ is the smallest integer
$h$ for which $u_h \neq 0$. In what follows, we denote it by $h(u)$.
When $h(u) = 0$, we say that $u$ is \emph{separable}.
When the characteristic $\p$ is zero, any 
isogeny is separable. On the contrary, when $\p$ does not vanish,
$h(u)$ is a necessarily a multiple of $\deg(\p)$, and $u$ decomposes as
$u = u_s \circ \tau^{h(u)}$,
where $\tau^{h(u)}$ defines an isogeny from $\phi$ to a second 
Drinfeld module $\phi'$ and $u_s : \phi' \to \psi$ is a \emph{separable} 
isogeny.

\subsubsection{Torsion points, Tate module, and Anderson motives}
\label{ssec:motive}

Let $\phi$ and $\psi$ be two rank $r$ Drinfeld modules. We define the most
important algebraic structures attached to a Drinfeld module.

\begin{deftn}[$A$-module]
  \begin{enumerate}[label=(\roman*)]
  \item The $A$-\emph{module} of $\phi$, denoted $\E(\phi)$, is the $A$-module
    $\Kbar$ equipped with the structure given by
    \[
      a \cdot z = \phi_a(z)
    \]
    for $a \in A$ and $z \in \E(\phi)$.
  \item Given an ideal $\a$ of $A$,
    we define the $\a$-torsion $\E_{\a}(\phi)$ of $\phi$ as the
    $\a$-\emph{torsion} of the module $\E(\phi)$, that
    is the subset of $\Kbar$ consisting of elements $z$ for which
    $\phi_a(z) = 0$ for all $a \in \a$.
    For an element $a \in A$, we write $\E_a(\phi)$
    for $\E_{aA}(\phi)$.
  \end{enumerate}
\end{deftn}

Any morphism of Drinfeld modules $u : \phi \to \psi$ induces
$A$-linear morphisms
\functiondefname
  {\E(u)}
  {\E(\phi)}
  {\E(\psi)}
  {z}
  {z(u)}
and
$\E_\a(u): \E_\a(\phi) \to \E_\a(\psi)$.
For any nonzero ideal
$\a \subset A$ away from the characteristic, the module $\E_\a(\phi)$ is free of rank $r$ over $A/\a$, \emph{i.e.}
$\E_\a(\phi) \simeq (A/\a)^r$ \cite[Remark~4.5.5.1]{gos98}. This classical fact
highlights one of the first similarities with elliptic curves, of which
rank two Drinfeld modules are said to be function field analogues.

\begin{deftn}[Tate module]
  Let $\q$ be a maximal ideal of $A$, away from the characteristic.
  We define the $\q$-adic \emph{Tate module} of $\phi$ as the inverse
  limit
  \[
    \T_\q(\phi) = \varprojlim \E_{\q^{hn}}(\phi)
  \]
  where $h$ is positive integer such that $\q^h$ is principal
  (which always exists because $A$ has finite class number), the
  transition map $\E_{\q^{h(n+1)}}(\phi) \to \E_{\q^{hn}}(\phi)$
  being given by $\phi_a$ where $a$ is a generator of $\q^h$.
\end{deftn}

\begin{rem}
The Tate module $\T_\q(\phi)$ does not depend, up to isomorphism,
on the choice of $h$, nor on the choice of a generator of $\q^h$.
\end{rem}

The Tate module $\T_\q(\phi)$ is a module over the completion 
$A_\q$ of $A$ with respect to the place $\q$.
It is free of rank $r$, and morphisms $u: \phi \to \psi$ give rise 
to $A_\q$-linear maps $\T_\q(u): \T_\q(\phi) \to \T_\q(\psi)$.

\begin{deftn}[Anderson motive]
\begin{enumerate}[label=(\roman*)]
\item
  The $A$-\emph{motive} of $\phi$, denoted by $\M(\phi)$, is the $A_K$-module $\Ktau$ 
  equipped with the structure given by
  \[
    (\lambda \otimes a) \cdot f = \lambda f \phi_a
  \]
  where $\lambda \in K$, $a \in A$, $f \in \M(\phi)$ and the
  multiplication in the right hand side is computed in $\Ktau$.
\item
  Given in addition an ideal $\a$ of $A$, we define
  \[
    \M_{\a}(\phi) = A/\a \otimes_A \M(\phi) = \M(\phi)/\a \M(\phi).
  \]
  For an element $a \in A$, we write $\M_a(\phi)$
  for $\M_{aA}(\phi)$.
\end{enumerate}
\end{deftn}

\begin{rem}
\label{rem:tauaction}
In classical references (\emph{e.g.} \cite[Section 5.4]{gos98}), the 
$A$-motive $\M(\phi)$ carries more structure: it is a module over
the noncommutative ring $K\{\tau\} \otimes_{\Fq} A = A_K\{\tau\}$. 
This additional $\tau$-action is important, but never
used in this article. Therefore, for simplicity, we only retain the structure of $A_K$-module.
\end{rem}

It is well known that $\M(\phi)$ is projective of rank $r$ over
$A_K$ (see \cite[Lemma~5.4.1]{gos98}). 
When $A = \Fq[T]$, we have $A_K \simeq K[T]$ and $\M(\phi)$ is 
free with basis $(1, \tau, \dots, \tau^{r-1})$ \cite[Lemma~3.4.4]{papikian_drinfeld_2023}.
We stress that this has significant importance for our algorithmic purpose.
In general, a morphism of Drinfeld modules $u : \phi \to \psi$ induces
a morphisms of $A_K$-modules
\functiondefname
  {\M(u)}
  {\M(\psi)}
  {\M(\phi)}
  {f}
  {fu}
and
$\M_\a(u): \M_\a(\psi) \to \M_\a(\phi)$.
We refer to \cite[Ch.~5]{gos98} or \cite[Section~2]{weil-pairing} for 
more details and generalizations. The degree of the Ore polynomial defining an
element $f \in \M(\phi)$ (resp. $\M(u)$) is called the $\tau$-degree of $f$
(resp. $\M(u)$).

\begin{rem}  
  Let $\a$ and $\q$ be ideals of $A$, with $\q$ maximal.
  The constructions $\E$, $\E_\a$, $\T_\q$, $\M$ and $\M_\a$ define 
  functors from the category of Drinfeld modules:
  \begin{itemize}[itemsep=0.5ex,parsep=0ex,topsep=0.5ex]
    \item $\E$ (resp. $\E_\a$) is a covariant functor to the category of
      $A$-modules (resp. $A/\a$-modules);
    \item $\T_\q$ is a covariant functor to the category of $A_\q$-modules;
    \item $\M$ (resp. $\M_\a$) is a contravariant functor to the category of
      $A_K$-modules\footnote{More precisely, $\M$ is a functor to the category of
      Anderson motives.} (resp. $A_K/\a A_K$-modules).
  \end{itemize}
  In standard references, the $\a$-torsion is denoted by $\phi[\a]$.
  In this article, we prefer the notation $\E_\a(\phi)$ because it 
  better underlines
  the functorial properties of the construction, which will later play a
  leading role.
\end{rem}

\subsubsection{Norms and characteristic polynomials}
\label{ssec:norm}

The norm of an isogeny is defined in \cite[\S 3.9]{gek91}, in terms of \emph{Euler-Poincaré 
characteristic}.
Let us take a step back, and fix a Dedekind domain $\A$.
The 
Euler-Poincaré characteristic, denoted by $\chi_\A$, is a function defined on
the class of finitely generated $\A$-modules and assuming
values in the set of ideals of $\A$. It is uniquely determined
by the following conditions:
\begin{enumerate}[label=(\roman*)]
\item
$\chi_\A(\A/\a) = \a$ for every ideal $\a$ of $\A$;
\item
$\chi_\A(M_2) = \chi_\A(M_1) \cdot \chi_\A(M_3)$
for every exact sequence $0 \to M_1 \to M_2 \to M_3 \to 0$ of finitely
generated $\A$-modules.
\end{enumerate}
The formation of Euler-Poincaré characteristic commutes with flat scalar extension. 
In particular, given a finitely generated $\A$-module $M$ and a maximal
ideal $\q \subset \A$, we have
\[
  \chi_\A(M) \otimes_\A \Aq = \chi_{\Aq}(M \otimes_\A \Aq).
\]
Similarly, if $\A'$ is another Dedekind domain lying above $\A$, we have
\[
  \chi_\A(M) \otimes_\A \A' = \chi_{\A'}(M \otimes_\A \A').
\]
If $M$ is torsion, the Noether's theorem on the structure of finitely generated
modules over Dedekind domains~\cite[Exercise~19.6]{eisenbud} implies that $M$
decomposes as $M \simeq \A/\a_1 \times \cdots \times \A/\a_\ell$, where
$\a_1, \ldots, \a_\ell$ are ideals of $\A$. In that case, $\chi_\A(M) =
\a_1\cdots \a_\ell$.

\begin{deftn}[Norm]

  Let $u: \phi \to \psi$ be an isogeny. The norm of $u$, denoted by
  $\norm(u)$, is defined~as
  \[
    \norm(u) = \p^{\frac{h(u)}{\deg(\p)}} \cdot \chi_A(\ker \E(u)).
  \]

\end{deftn}

\begin{rem}
  We recall that $h(u)$ denotes the height of $u$.
  This definition takes into account that an isogeny and its 
  separable part have the same kernel: the correction by the factor 
  $\p^{h(u)/\deg(\p)}$ corresponds to the purely inseparable part.
\end{rem}

\begin{ex}
  Let $r$ be the rank of $\phi$.
  For $a \in A$, we have $\norm(\phi_a) = a^r A$.
  If $\p \neq 0$ then $\norm(\tau^{\ell \deg(\p)}) = \p^\ell$ for 
  all $\ell \in \NN$. In particular, when $K$ is a
  finite extension of degree $d$ of $\Fq$, the norm of the
  Frobenius endomorphism $F_\phi$ is explicitly given by
  $\norm(F_\phi) = \p^{d/\deg(\p)}$.
\end{ex}

One proves~\cite[Lemma~3.10]{gek91} that the norm is multiplicative: 
if $u$ and $v$ are composable isogenies, we have
$\norm(v\circ u) = \norm(v) \cdot \norm(u)$.
When $u$ is an endomorphism, its action on the Tate module $\T_\q(u)$ 
is a linear endomorphism, whose determinant lies in $A$ and generates 
$\norm(u)$ \cite[Lemma~3.10.iii]{gek91}:
\[
  \norm(u) = \det(\T_\q(u)) \cdot A.
\]

\begin{deftn}[Characteristic polynomial]

  Let $u: \phi \to \phi$ be an endomorphism.
  We define the \emph{characteristic polynomial} of $u$ as the
characteristic polynomial of $\T_q(u)$.

\end{deftn}

Since $\T_{q}(\phi)$
has rank $r$ over $A_\q$, the characteristic polynomial of $u$ has degree $r$.
It is also proven that it has coefficients in $A$ \cite[Corollary~3.4]{gek91}.

\begin{ex}
  In this example, we assume that $A = \Fq[T]$, that $K$ is finite of degree $d$ over
  $\Fq$, and that $\phi$ is a rank two Drinfeld module defined by
  $\phi_T = \gamma(T) + g \tau + \Delta \tau^2$.
  The characteristic polynomial of the Frobenius endomorphism of $\phi$
  takes the form~\cite[Theorem~2.11]{gek08}
  \[
    X^2 \,-\, tX \,+\, (-1)^d \text{N}_{K/\Fq}(\Delta)^{-1} \p^{d/\deg(\p)}
  \]
  where $\text{N}_{K/\Fq}$ is the norm from $K$ to $\Fq$ and, in a slight 
  abuse of notation, the notation $\p$ is used to denote the monic generator of
  the characteristic.
  The coefficient $t \in \Fq[T]$ is called the \emph{Frobenius 
  trace} of $\phi$ and we have $\deg_T(t) \leq d/2$. We refer to
  Remark~\ref{rem:frobenius-norm} for more
  information about the \emph{Frobenius norm}.
  The endeavour of computing this polynomial has been the object of many
  research articles, leading to a variety of algorithms. We refer to Appendix~\ref{appendix:review} for a review of
  their respective complexities.
\end{ex}

\subsubsection{Restriction of Drinfeld modules}
\label{sssec:restriction}

We consider $\gamma' : A' \to K$, a second
base for Drinfeld modules satisfying the assumptions
of \S \ref{ssec:backgrounddrinfeld}, and we assume that we are given in
addition an injective homomorphism of rings $f : A' \to A$ such
that $\gamma' = \gamma \circ f$.
Thanks to our assumptions on $A$ and $A'$, we find that $f$ endows
$A'$ with a structure of finite $A$-algebra.
If $\phi : A \to K\{\tau\}$ is a Drinfeld module, the composite
\[
  \phi \circ f : A' \to A \to K\{\tau\}
\]
defines a Drinfeld module over $A'$, denoted by $f^* \phi$ and referred to as
the \emph{restriction} of $\phi$ along~$f$.

Considering two Drinfeld $A$-modules as well as a morphism $u : \phi \to \psi$,
one checks that the Ore polynomial defining $u$ also defines
an isogeny $f^* \phi \to f^* \psi$, which we denote by $f^* u$. 
The construction $f^*$ defines a functor from the category
of Drinfeld modules over $A$ to the category of Drinfeld modules
over $A'$.
The action of $f^*$ on the motives is easy to describe: the motive
$\M(f^* \phi)$ is simply $\M(\phi)$ with the restricted action of
$A$ and, for any morphism $u: \phi \to \psi$, the maps $\M(f^*u)$ 
and $\M(u)$ are the same (up to the above identification).

\subsection{Algorithmics}

We now move to algorithmics and discuss the complexity of
performing basic operations on matrices on the one hand, and on
Ore polynomials on the other hand.

\subsubsection{Complexity model}
\label{sssec:complexitymodel}

We recall the Landau's notation $O$, $\Otilde$ and $\Opower$
from the introduction: if $f$ and $g$
are two positive quantities depending on parameters, we write
\begin{itemize}[itemsep=0.5ex,parsep=0ex,topsep=0.5ex]
\item $g \in O(f)$ if there exists an absolute positive constant $C$
such that $g \leq C {\cdot} f$ for all choices of parameters,
\item $g \in \Otilde(f)$ if there exist absolute positive constants $C$
and $k$ such that $g \leq C {\cdot} f \log^k f$ for all choices of 
parameters,
\item $g \in \Opower(f)$ if, for all $\varepsilon > 0$, there exists
a positive constant $C_\varepsilon$ such that $g \leq C_\varepsilon
{\cdot} f^{1 + \varepsilon}$ for all choices of parameters.
\end{itemize}
We notice that $O(f) \subset \Otilde(f) \subset \Opower(f)$ for
all $f$ as above. Moreover, if $f_1$ and $f_2$ are two quantities
as above, one checks that
$O(f_1) + O(f_2) \subset O(f_1 + f_2)$,
$\Otilde(f_1) + \Otilde(f_2) \subset \Otilde(f_1 + f_2)$ and,
similarly, $\Opower(f_1) + \Opower(f_2) \subset \Opower(f_1 + f_2)$.

In this article, we measure complexity in two different ways.
When $K$ is an arbitrary field, we use \emph{arithmetic
complexity}, meaning that we count separately arithmetic
operations (addition, subtraction, multiplication and division)
in~$K$ on the one hand, and applications of Frobenius (that is the
computation of $x^q$ for a given $x \in K$) on the
other hand.

On the contrary, when $K$ is a finite field, we rather use 
\emph{bit complexity}, meaning that we count operations on 
bits.
When $K$ is a finite extension of $\Fq$ of degree~$d$ presented as 
a quotient $K = \Fq[X]/Q(X)$ (for some irreducible polynomial $Q(X)
\in \Fq[X]$ of degree~$d$) and when $\Fq$ is itself presented as a 
quotient of $\Fp[X]$, classical algorithms based on Fast Fourier 
Transform allows for performing all arithmetic operations in~$K$ 
for a cost of $\Otilde(d \log q)$ bit operations (see for instance
\cite[Chapter~II]{gathen}).

Estimating the cost of applying the Frobenius endomorphism of~$K$ 
is more challenging, even though partial results are
available in the literature. First of all, Kedlaya and Umans'
algorithm~\cite{kedlaya-umans} for fast modular composition is 
theoretically capable to compute an image by Frobenius for a cost of 
$\Opower(d \log q)$ bit operations. However, if $\alpha$ denotes the 
image of $X$ in $K$, one needs nevertheless to precompute $\alpha^q$,
\emph{i.e.} to write $\alpha^q$ on the
canonical monomial basis $(1, \alpha, \ldots, \alpha^{d-1})$. Using a fast
exponentiation algorithm, this can be done for an initial cost of 
$\Otilde(d \log^2 q)$ bit operations. Another
flaw with this approach is that, as far as we know, one still 
lacks an efficient implementation of Kedlaya and Umans' algorithm.

Another option, which achieves quasi-optimal complexity, is to use 
the elliptic normal bases of Couveignes and Lercier~\cite{couveignes-lercier}
instead of the classical monomial basis.
Indeed, in those bases, all arithmetic operations and applications
of Frobenius can be computed for a cost of $\Otilde(d)$ operations
in $\Fq$, corresponding to $\Otilde(d \log q)$ bit operations.
The drawback of this solution is that constructing an elliptic
normal basis can be costly. Nevertheless this needs to be done only once, at the
instantiation of $K$.

Taking all of this into account, we choose to follow the convention of~\cite{musleh-schost-2}
and opt for the first option: we make the 
assumption that all arithmetic operations and applications of 
Frobenius in~$K$ cost $\Opower (d \log q)$ bit operations, plus
a unique initial cost of $\Otilde(d \log^2 q)$ operations
for the precomputation of~$\alpha^q$.

\subsubsection{Polynomial matrices}
\label{sec:matrices-computations}

We give a rough review of the literature on the computation of determinants and
characteristic polynomials of polynomial matrices. We recall from the
introduction that the notation $\omega \in [2,3]$ refers to feasible exponent
for matrix multiplication. When matrices have coefficients in a field $L$, both
computing determinants and characteristic polynomials reduce to matrix
multiplication~\cite{matrices:field-charpoly-mult, matrices:field-charpoly}.
Computing the determinant of a polynomial matrix also reduces to matrix multiplication~\cite{matrices:poly-computations-03,
matrices:poly-computations-05}. However, the situation of the characteristic
polynomial is more delicate. Consider a $s$-by-$s$ matrix with entries in
$L[T]$. Computing its characteristic polynomial can be done for a cost of
$\Otilde(s^\Omega n)$ operations in $L$ with $\Omega < 2.69497$
\cite{matrices:kaltofen, matrices:kaltofen-villard}.

When $M$ is a $s$-by-$s$ matrix, we use the notation $\pi(M)$ to be to
its monic characteristic polynomial, that is $\pi(M) = \det(X{\cdot}I_s - M)$
where $I_s$ is the identity matrix of size $s$. In the next
two lemmas, we derive two useful algorithms, for two specific situations.

\begin{lem}\label{lem:matrice-comp-1}
  We assume that $L$ is a finite field of degree $d$ over $\Fq$.
  Let $M$ be a $s$-by-$s$
  matrix with coefficients in $L[T]$.
  Let~$n$ be a uniform upper bound on the degree of the coefficients of
  $\pi(M)$.
  There exists a Las Vegas algorithm that computes
  the $\pi(M)$ for a cost of 
  $\Opower(n/d) + \Otilde((n {+} d)s^\omega)$ operations in~$\Fq$.
\end{lem}

\begin{proof}
Let $L'$ be an extension of $L$ of degree $\lceil n / d \rceil$; such
an extension, altogether with a generator $\alpha$ of $L'$ over $\Fq$, can be found out using Couveignes and Lercier's Las Vegas 
algorithm, whose complexity is in $\Opower(\frac n d)$ operations
in~$\Fq$~\cite{couveignes-lercier-2}.
The degree of the extension $L'/\Fq$ is then in the range $[n, n{+}d]$. 
Let $M(\alpha)$ denote the evaluation of $M$ at $T = \alpha$, and
write its characteristic polynomial as follows:
\[
  \pi(M(\alpha)) = \sum_{i=0}^s \sum_{j=0}^n a_{i, j} \alpha^i X^i.
\]
where the coefficients $a_{i,j}$ are in $\Fq$. Then
\[
  \pi(M) = \sum_{i=0}^s \sum_{j=0}^n a_{i, j} T^i X^i.
\]
The generator $\alpha$ being known, computing $\pi(M(\alpha))$ costs $\Otilde(s^\omega)$
operations in $L'$, which corresponds to $\Otilde((n{+}d) s^\omega) $
operations in $\Fq$.
\end{proof}

\begin{lem}\label{lem:matrice-comp-2}
  Let $M$ be a
  $s$-by-$s$ matrix with coefficients in $\Fq[T]$ and let $n$ be
  a uniform upper bound on the degrees of the entries of $M$.
  We assume that the coefficients of $\pi(M)$ fall in $\Fq[T^s]$.
  There exists a Las Vegas algorithm that computes $\pi(M)$ with
  probability of success at least $\frac 1 2$ for a cost of
  $\Otilde(n s^\omega)$ operations in $\Fq$.
\end{lem}

\begin{proof}
Let $\alpha_1, \dots, \alpha_n \in \Fq$ be such that $\alpha_i^s \neq 
\alpha_j^s$ whenever $i \neq j$. 
We compute the matrices $M(\alpha_1), \ldots, M(\alpha_n)$ 
and compute their characteristic polynomials 
$\pi(M(\alpha_1)), \ldots, \pi(M(\alpha_n))$, for a total cost of
$\Otilde(ns^\omega)$ operations in~$\Fq$. Thanks to our 
assumption, $\pi(M)$ can be seen as having $s$ polynomial coefficients 
of degree at most $n$. Using fast interpolation algorithms~\cite[\S II.10]{gathen},
$\pi(M)$ can therefore 
be recovered from the $\pi(M(\alpha_i))$'s for a cost of $\Otilde(ns)$ 
operations in $\Fq$. We end up with a total of $\Otilde(n s^\omega)$ 
operations in $\Fq$.

This procedure only works if $\Fq$ is large enough to 
pick a valid set $\{\alpha_1, \dots, \alpha_n\}$.
Let $\rho = \frac{\gcd(q-1, s)}{q-1}$ be the proportion of 
elements in $\FF_q^\times$ that are $d$-th roots of unity. A family 
$(\alpha_1, \dots, \alpha_n) \in (\FF_q^\times)^n$ has probability 
$p_n = (1 - \rho)(1 - 2\rho)\cdots (1 - n \rho)$ to form a valid set. As 
$p_n \geq 1 - \frac{n(n+1)}{2} \rho$, the process has a chance of 
success greater than $\frac{1}{2}$ as soon as $q > 1 + s n (n+1)$. 
If $\Fq$ is not large enough, we do all computations in a finite 
extension of $\Fq$. With these estimations, we conclude that it is enough to work in an 
extension whose degree has order of magnitude $\log_q(s n^2)$. Building 
this extension, as well as computing in it, does not affect the 
announced complexity.
\end{proof}

\subsubsection{Ore polynomials}
\label{subsec:ore-computations}

In full generality,
multiplications and Euclidean divisions of Ore polynomials in 
$K\{\tau\}$ of degree at most $n$ can be achieved with the naive 
algorithm for a cost of $O(n^2)$ operations in $K$ and $O(n^2)$ extra 
applications of the Frobenius endomorphism.

However, when $K$ is a finite field, we can take advantage of fast Ore polynomial 
multiplication~\cite{caruso-leborgne-1, caruso-leborgne-2}. As before,
we use the letter~$d$ to denote the degree of the extension $K/\Fq$.
Let $\SM(n,d)$ denote a function having the following property: the 
number of bit operations needed for multiplying two 
Ore polynomials in $\Ktau$ of degree less than $n$ is in
$\Opower(\SM(n,d)\log q)$. 
At the time of writing this article, the best known value of $\SM$ is
given in \cite{caruso-leborgne-2}\footnote{In~\cite{caruso-leborgne-2}, 
the complexity is given in number of
operations in the ground field~$\Fq$, with the assumption that 
applying the Frobenius endomorphism of~$K$ requires at most 
$\Otilde(d)$ operations in~$\Fq$. Consequently one operation in
$\Fq$ in the setting of~\cite{caruso-leborgne-2} corresponds to
$\Opower(\log q)$ bit operations in the complexity model of this
article (see \S \ref{sssec:complexitymodel}).}\textsuperscript{,}%
\footnote{Note that there is a typo in \cite{caruso-leborgne-2}: the 
critical exponent is not $\frac{5-\omega}2$ but $\frac 2{5-\omega}$.}:

\begin{center}
  \begin{tabular}{r@{\hspace{0.5ex}}ll}
    $\SM(n, d)$ & $= n^{\frac{\omega+1}{2}} d$
         & for $n \leq d^{\frac 2 {5 - \omega}}$, \\
         & $= n^{\omega-2} d^2$
         & for $d^{\frac 2 {5 - \omega}} \leq n \leq d$, \\
         & $= n d^{\omega - 1}$
         & for $d \leq n$.
  \end{tabular}
\end{center}

\noindent
Let also $\SMgeq$ be the function defined by
\[
  \SMgeq(n, d) = \sup_{0 < m \leqslant n} \SM(m, d) \frac{n}{m}.
\]
The function $\SMgeq$ is the smallest log-concave function above $\SM$.
It is proved in \cite{caruso-leborgne-2} that computing the 
right-Euclidean division of Ore polynomials in $\Ktau$ of degree less 
than $n$ requires at most $\Opower(\SMgeq(n, d)\log q)$ bit operations.
With the above values for $\SM(n,d)$, we have

\begin{center}
  \begin{tabular}{r@{\hspace{0.5ex}}ll}
    $\SMgeq(n, d)$ & $= n^{\frac{\omega+1}{2}} d$
               & for $n \leq d^{\frac 2{5 - \omega}}$, \\
               & $= n d^{\frac{4}{5 - \omega}}$
               & for $d^{\frac 2 {5 - \omega}} \leq n$.
  \end{tabular}
\end{center}

\section{Characteristic polynomials of endomorphisms}
\label{sec:endomorphisms}

In this section, we recall that characteristic polynomials
of \emph{endomorphisms} of Drinfeld modules can be read off at the
level of Anderson motives. We then take advantage of this motivic
interpretation to design fast algorithms (including the algorithms
\FMFF and \FMKU mentioned in the introduction) for computing Drinfeld module
endomorphism characteristic polynomials.

\subsection{Duality between torsion points and $A$-motives}

It is a standard result in the theory of Drinfeld modules that 
$A$-motives are duals to the so-called $A$-modules which, in some sense, 
correspond to torsion points (see for instance \cite[Sections~5.4, 
5.6]{gos98} or \cite[\S 3.6]{papikian_drinfeld_2023}). We hereby propose 
a concrete incarnation of this yoga, establishing a duality between the 
functors $\E_\a$ and $\M_\a$. The material presented in this subsection 
is somehow classical. However, we believe that our presentation is more elementary than 
those from aforementioned references: for 
instance, we do not need the introduction of (abelian) $A$-modules. As such, we
include all proofs, hoping they will be of interest for some readers.

Let $\a$ be an ideal of $A$ away from the characteristic. We consider the
evaluation map
\functiondef
  {\B : \quad \E(\phi) \times \M(\phi)}
  {\Kbar}
  {(z, f)}
  {f(z).}
It is easily checked that $\B$ is $\Fq$-linear with respect to
the variable $z$ and $K$-linear with respect to the variable
$f$. Moreover, it follows from the definitions that $\B$ vanishes
on the subset $\E_\a(\phi) \times \a \M(\phi)$ and therefore
induces a bilinear mapping
\functionname
  {\B_\a}
  {\E_\a(\phi) \times \M_\a(\phi)}
  {\Kbar.}
We consider the scalar extensions $\E_\a(\phi)_{\Kbar} = \Kbar \otimes_{\Fq}
\E_\a(\phi)$ and $\M_\a(\phi)_{\Kbar} = \Kbar \otimes_K \M_\a(\phi)$. The map
$\B_\a$ induces a $\Kbar$-bilinear form
\functionname
  {\B_{\a,\Kbar}}
  {\E_\a(\phi)_{\Kbar} \times \M_\a(\phi)_{\Kbar}}
  {\Kbar.}

\begin{prop}\label{prop:pairing}
  The bilinear form $\B_{\a,\Kbar}$ is a perfect pairing.
\end{prop}

\begin{proof}
  Recall that, since $\a$ is away from the characteristic, 
  $\E_\a(\phi)$ is free with rank $r$ over $A/\a$. Therefore,
  $\dim_{\Fq} \E_\a(\phi) = r \cdot \deg(\a) = \dim_K \M_\a(\phi)$,
  and $\E_\a(\phi)_{\Kbar}$ and $\M_\a(\phi)_{\Kbar}$ have the same
  dimension over $\Kbar$.

  It is then enough to prove that $\B_{\a,\Kbar}$ is nondegenerate on the left,
  meaning that if $x \in \E_\a(\phi)_{\Kbar}$ satisfies $\B_{\a,\Kbar} (x,y) =
  0$ for all $y \in \M_\a(\phi)_{\Kbar}$, then $x$ must vanish. More
  generally, we are going to prove that there is no nonzero $x \in
  \E_\a(\phi)_{\Kbar}$ having the following property: 
  $\B_{\a,\Kbar}(x, 1 \otimes \tau^j) = 0$ for all $j$ large enough. 
  We argue by contradiction and
  consider an element $x \in \E_\a(\phi)_{\Kbar}$ satisfying the above
  property. We write
  \[
    x = \lambda_1 \otimes z_1 + \cdots + \lambda_n \otimes z_n.
  \]
  with $\lambda_i \in \Kbar$ and $z_i \in \E_\a(\phi)$. Moreover, we assume
  that $x$ is chosen in such a way that the number of terms $n$ is minimal.
  This ensures in particular that the $z_i$'s are linearly independent over
  $\Fq$. Writing that $\B_{\a,\Kbar}(x, 1 \otimes \tau^j)$ vanishes, we obtain
  the relation
  \[
    (E_j) : \quad
    \lambda_1 z_1^{q^j} + \cdots + \lambda_n z_n^{q^j} = 0,
  \]
  which, in turn, implies
  \[
    (E'_j) : \quad
    \lambda_1^q z_1^{q^{j+1}} + \cdots + \lambda_n^q z_n^{q^{j+1}} = 0.
  \]
  Combining the relations $(E_{j+1})$ and $(E'_j)$, we find
  \[
    (\lambda_1^q - \lambda_n^{q-1} \lambda_1) \cdot z_1^{q^{j+1}} + \cdots + 
    (\lambda_{n-1}^q - \lambda_n^{q-1} \lambda_{n-1}) \cdot z_{n-1}^{q^{j+1}} = 0.
  \]
  In other words, the vector
  \[
    y = (\lambda_1^q - \lambda_n^{q-1} \lambda_1) \otimes z_1^{q^{j+1}} + \cdots +
    (\lambda_{n-1}^q - \lambda_n^{q-1} \lambda_{n-1}) \otimes z_{n-1}^{q^{j+1}}
    \in \E_\a(\phi)_{\Kbar}
  \]
  is a new solution to our problem.

  This will contradict the minimality condition in the choice of $x$ if we can
  prove that $y$ does not vanish. To do this, we again argue by contradiction.
  Given that the $z_i$'s are linearly independent over $\Fq$, the vanishing of
  $y$ would imply $\lambda_i^q - \lambda_n^{q-1} \lambda_i = 0$ for all
  $i$, from which we would deduce that all the quotients
  $\frac{\lambda_i}{\lambda_n}$ lie in $\Fq$. Thanks to the relations $(E_j)$,
  this again contradicts the linear independence of the $z_i$'s over $\Fq$.
\end{proof}

\begin{rem}
  Proposition~\ref{prop:pairing} can be seen as a Drinfeld analogue of the
  classical pairing between the singular homology and the de Rham cohomology of
  a complex abelian variety: the space $\E_\a(\phi)$ plays the role of the
  singular homology (\emph{via} the étale viewpoint), while the space
  $\M_\a(\phi)$ can be thought of as the incarnation of the de Rham cohomology
  (see~\cite{angles}).
\end{rem}

Proposition~\ref{prop:pairing} gives a natural identification
\[
  \alpha_\phi :
  \quad
  \E_\a(\phi)_{\Kbar} 
  \simeq \Hom_{\Kbar}\big(\M_\a(\phi)_{\Kbar}, \Kbar\big)
  \simeq \Hom_K\big(\M_\a(\phi), \Kbar\big),
\]
where $\Hom_{\Kbar}$ (resp. $\Hom_K$) refers to the space of $\Kbar$-linear
(resp. $K$-linear) morphisms. \emph{A priori}, the isomorphism $\alpha_\phi$ is
only $\Kbar$-linear; we upgrade it and make it $A_{\Kbar}$-linear.

\begin{deftn}
  Let $M$ be a module over $A_K$. We set $M^\ast = \Hom_K(M, K)$ and equip it
  with the structure of $A_K$-module given by
  \[
    a \cdot \xi = \big(m \mapsto \xi(am)\big),
  \]
  where $a \in A_K$ and $\xi \in M^\ast$.
\end{deftn}

One checks that the construction $M \mapsto M^\ast$ is functorial, in the sense
that if $g : M_1 \to M_2$ is a morphism of $A_K$-modules, then the dual map
$g^\ast : M_2^\ast \to M_1^\ast$ is $A_K$-linear as well. We define
$\M_\a(\phi)^\ast_{\Kbar} = \Kbar \otimes_K \M_\a(\phi)^\ast$; it is a module
over $A_{\Kbar}$.
A direct adaptation of \cite[Lemma~3.6.2]{papikian_drinfeld_2023} 
using Noether's structure theorem for finitely generated modules over a 
Dedekind domain~\cite[Theorem~A3.2]{eisenbud} gives the following lemma.

\begin{lem}
\label{lem:isomorphisme-dual}
Any torsion finitely generated $A_K$-module $M$ is
(noncanonically) isomorphic to its dual~$M^\ast$.
\end{lem}

\begin{theo}\label{th:pairing}
  The perfect pairing $\B_{\a,\Kbar}$ induces an $A_{\Kbar}$-linear
  isomorphism:
  \[
    \alpha_\phi :
    \quad \E_\a(\phi)_{\Kbar} 
    \;
    \stackrel{\sim}{\longrightarrow}
    \;
    \M_\a(\phi)^\ast_{\Kbar}.
  \]
  Moreover, given a Drinfeld module morphism $u : \phi \to \psi$,
  the following diagram is commutative:
  \[\begin{tikzcd}[column sep=huge,row sep=large]
    {\E_\a(\phi)_{\Kbar}} & {\E_\a(\psi)_{\Kbar}} \\
    {\M_\a(\phi)^\ast_{\Kbar}} & {\M_\a(\psi)^\ast_{\Kbar}}
    \arrow["{\id \otimes \E_\a(u)}", from=1-1, to=1-2]
    \arrow["\alpha_\phi"', from=1-1, to=2-1]
    \arrow["{\id \otimes \M_\a(u)^\ast}"', from=2-1, to=2-2]
    \arrow["\alpha_\psi", from=1-2, to=2-2]
  \end{tikzcd}\]
\end{theo}

\begin{proof}
  For the first assertion, we already know that $\alpha_\phi$ is a
  $\Kbar$-linear isomorphism. It then only remains to verify that it is
  $A$-linear. Let $a \in A$ and $z \in \E_\a(\phi)$. By definition $a{\cdot}z =
  \phi_a(z)$ and $a{\cdot}f = f \phi_a$ for $f \in \M(\phi)$. Hence
  $\alpha_\phi(a{\cdot}z)$ is the function $f \mapsto 
    f\big(\phi_a(z)\big) = (f \phi_a) (z) = (a{\cdot}f)(z)$,
  which means that $\alpha_\phi(a{\cdot}z) = a {\cdot} \alpha_\phi(z)$
  as desired.
  The second assertion is easily checked.
\end{proof}

\begin{rem}
  Theorem~\ref{th:pairing} shows that $\E_\a(\phi)_{\Kbar}$ determines
  $\M_\a(\phi)_{\Kbar}$ and \emph{vice versa}. One can actually do much better
  and obtain a direct correspondence between $\E_\a(\phi)$ and
  $\M_\a(\phi)$ without extending scalars to $\Kbar$
  (see, for instance, \cite[Equation~(3.6.9)]{papikian_drinfeld_2023}).
  For this, we need to add
  more structures. On the one hand, on $\M_\a(\phi)$, we retain the
  $\tau$-action as discussed in Remark~\ref{rem:tauaction}. On the other hand,
  on $\E_\a(\phi)$, we have a Galois action. Precisely let $\Ksep$ denote the
  separable closure of $K$ inside $\Kbar$. From the fact that $\a$ is away from
  the characteristic, we deduce that $\E_\a(\phi)$ lies in $\Ksep$, and endow
  with an action of the Galois group $G_K = \Gal(\Ksep/K)$. We now have the following
  identifications refining those of Theorem~\ref{th:pairing}:
  \begin{align*}
    \E_\a(\phi) & \simeq \Hom_{\Ktau}\big(\M_\a(\phi), \Ksep\big) \\
    \M_\a(\phi) & \simeq \Hom_{\Fq[G_K]}\big(\E_\a(\phi), \Ksep\big)
  \end{align*}
  where, in the first (resp. second) line, we consider
  $K$-linear morphisms commuting with the $\tau$-action (resp. $\Fq$-linear
  morphisms commutating with the Galois action). In other words, the Galois representation
  $\E_\a(\phi)$ and the $\tau$-module $\M_\a(\phi)$
  correspond one to the other under Katz' anti-equivalence of
  categories~\cite[Proposition~4.1.1]{katz}.
\end{rem}

\begin{rem}
  In \cite{weil-pairing}, van der Heiden proposes another approach,
  proving that there is a canonical $A$-linear isomorphism:
  \[
    \E_\a(\phi) \simeq
    \Hom_{A/\a}\big(\M_\a(\phi)^\tau, \Omega_A / \a \Omega_A\big)
  \]
  where $\M_\a(\phi)^\tau$ denotes the subset of fixed points of
  $\M_\a(\phi)$ by the $\tau$-action and
  $\Omega_A$ is the module of Kähler differential forms
  of $A$ over $\Fq$ (see Proposition~4.3 of \emph{loc.~cit.}).
  However, the formulation of Theorem~\ref{th:pairing} is better
  suited for the applications we shall develop in this article.
\end{rem}

If $M$ is a finitely generated projective $A_K$-module of rank $n$, 
we let
\[
  \det M = \bigwedge^n M
\]
denote the maximal exterior power of $M$.
Any $A_K$-linear endomorphism $f : M \to M$ induces a linear map $\det f : 
\det M \to \det M$. The latter is
the multiplication by some element of $A_K$, that we call the
\emph{determinant} of $f$ and denote by 
$\det f$ in a slight abuse of notation.
Similarly, we define the characteristic polynomial of~$f$ as the
determinant of the $A_K[X]$-linear map $X{-}f$ acting on $A_K[X]
\otimes_{A_K} M$.

A classical consequence of Theorem~\ref{th:pairing} is the following.

\begin{theo}\label{th:norm-endomorphism}
  Let $\phi$ be a Drinfeld module and 
  let $u : \phi \to \phi$ be an endomorphism.
  Let $\q \subset A$ be
  a maximal ideal away from the characteristic. Then
  the characteristic polynomials of $\T_\q(u)$ and $\M(u)$ are equal.

  In particular, $\norm(u)$ is the principal ideal generated by $\det(\M(u))$.
\end{theo}

\begin{proof}
  Let $n \in \NN$. Applying Theorem~\ref{th:pairing} with
  $\a = \q^n$, we find
  \[ 
    \pi\big(\E_{\q^n}(u)\big) 
  = \pi\big(\E_{\q^n}(u)_{\Kbar}\big)
  = \pi\big(\M_{\q^n}(u)_{\Kbar}\big) 
  = \pi\big(\M_{\q^n}(u)\big),
  \]
  the second equality being a consequence of Theorem~\ref{th:pairing} and
  the fact that two dual morphisms have the same determinant (in
  suitable bases, their matrices are transposed one to the other).
  Thus we obtain
  $\pi (\T_\q(u)) \equiv \pi (\M(u)) \pmod{\q^n}$.
  Since this holds for all positive integer $n$, we conclude that
  $\pi(\T_\q(u)) = \pi(\M(u))$.

  The last statement now
  follows from \cite[Lemma~3.10]{gek91}.
\end{proof}

\subsection{Algorithms: the case of $\mathbb P^1$}
\label{subsec:algo-motive}

In this subsection, we assume that $A = \Fq[T]$, and we
let $\phi$ be a Drinfeld module of rank $r$. We fix an endomorphism $u: \phi \to \phi$ 
and aim at designing an algorithm that computes
the characteristic polynomial (resp. norm) of~$u$.
Under the assumption that $A = \Fq[T]$, the ring $A_K \simeq K[T]$ is a principal ideal
domain and $\M(\phi)$
is free of rank~$r$. Moreover, a canonical basis is 
given by $(1, \tau, \dots, \tau^{r-1})$. Our strategy is then clear:
we compute the matrix representing the $K[T]$-linear map $\M(u)$ in
the aforementioned canonical basis and then return its characteristic
polynomial (resp. determinant); Theorem~\ref{th:norm-endomorphism} ensures that
it is
the characteristic polynomial (resp. norm) of~$u$.

\subsubsection{Generic algorithm}

Our first need is to design an algorithm for computing the coordinates of an element
$f \in \M(\phi)$, represented as an Ore polynomial, in the canonical
basis of $\M(\phi)$.
This is achieved by Algorithm~\ref{algo:motive-coordinates}, whose
correctness is immediately proved by induction on the $\tau$-degree
of~$f$.

\begin{algorithm}[h]
    \caption{\MotiveCoordinates}
    \label{algo:motive-coordinates}
    \KwIn{An element $f$ in the motive $\M(\phi)$}
    \KwOut{The coordinates $(f_0, \dots, f_{r-1})$ of $f$
    in the canonical basis of $\M(\phi)$}

    \eIf{$\deg f < r$}{
        \KwRet the vector defined by the coefficients of $f$ \;
    }{
        Set $m = \max(1, \lfloor \deg(f) / 2r \rfloor)$ \;
        Write $f = a \cdot \phi_X^m + b$ with $\deg (b) < rm$ (right Euclidean division)\;
        \KwRet $X^m \cdot \MotiveCoordinates(a) + \MotiveCoordinates(b)$ \;
    }
\end{algorithm}

\begin{lem}
\label{lemma:comp-motive-coordinates}
  For an input $f \in \M(\phi)$ of $\tau$-degree $n$,
  Algorithm~\ref{algo:motive-coordinates} requires $O(n^2)$ applications
  of the Frobenius endomorphism and $O (n^2)$ operations in $K$.
\end{lem}

\begin{proof}
  The first step of the algorithm consists in computing $\phi_T^m$.
  Using fast exponentiation, this costs $O (n^2)$ applications of the
  Frobenius endomorphism and $O (n^2)$ operations in $K$. The Euclidean
  division requires $O (n^2)$ applications of the Frobenius endomorphism
  and $O (n^2)$ operations in $K$ as well.

  Let $\mathrm{C}(s)$ be the cost of running the algorithm on an entry with
  degree $s$. By what precedes, $\mathrm{C}(s)$ is less than
  $\mathrm{C}(\lceil\frac{s}{2}\rceil)$, plus $O (s^2)$ operations in $K$ and
  $O (s^2)$ applications of the Frobenius endomorphism.
  We conclude using the Master
  Theorem~\cite[Theorem~4.1]{cormen_introduction_2022}.
\end{proof}

From Algorithm~\ref{algo:motive-coordinates}, we also derive the
following bounds on the size of the coefficients.

\begin{lem}\label{lemma:bound-coeffs}

  Let $f \in \M(\phi)$ and let $f_0, \ldots, f_{r-1} \in K[T]$ be the
  coordinates of $f$ in the canonical basis.
  Then for
  $0 \leq i < r$, either $\deg(f) < i$ and $f_i = 0$, or $\deg(f)
  \geqslant i$, in which case we have
  \[
    \deg_T(f_i) \leq \frac{\deg(f) - i} r.
  \]
\end{lem}

\begin{cor}\label{cor:bound-coeffs}
  Let $(P_{i, j})_{0 \leq i,j < r}$ be the matrix of $\M(u)$ in the canonical
  bases.
  Then for $0 \leq i, j \leq r-1$, either $\deg(u)
  + j < i$ and $P_{i, j} = 0$, or $\deg(u) + j \geqslant i$, in which
  case we have
  \[
    \deg(P_{i, j}) \leq \frac{(\deg(u) + j) - i} r.
  \]
\end{cor}

\begin{proof}
  By definition, $P_{i,j}$ is the coefficient in front of $\tau^i$
  in the decomposition of $\tau^j u$ in the canonical basis. 
  The corollary then follows from Lemma~\ref{lemma:bound-coeffs}.
\end{proof}

As a consequence of the previous statements, we obtain an alternative
proof of the following 
classical result~\cite[Theorem~4.2.7]{papikian_drinfeld_2023}.

\begin{prop}\label{prop:borne-coeff-frob-charpoly}
  We assume that $K$ is a finite field.
  Let $\pi = \pi_0(T) + \cdots + \pi_r(T) X^r$ be the characteristic polynomial
  of the Frobenius endomorphism of $\phi$. Then for every $0 \leqslant i
  \leqslant r$ we have
  \[
    \deg(\pi_i) \leqslant \frac {r - i} {r} d.
  \]
\end{prop}

\begin{proof}
  Using Theorem~\ref{th:norm-endomorphism}, we know that $\pi$ is the
  characteristic polynomial of the matrix $P$ of $\M(\tau^d)$ in the canonical
  bases. Therefore, for every $0 \leq i \leq r$, $\pi_i$ is the trace of
  $\bigwedge^i \M(\tau^d)$, which is an alternated sum on the principal minors of $P$ with
  size $i$. We conclude using Corollary~\ref{cor:bound-coeffs}.
\end{proof}

We now go back to our original setting; that is, $K$ and its function field
characteristic can be either finite or
infinite. Instead of independently computing all columns using
Algorithm~\ref{algo:motive-coordinates}, a more intelligent approach can be
employed to calculate the matrix of $\M(u)$:
in order to speed up the computation of a column, we may reuse those that are already
computed.
For this, we write
\[
  \phi_T = g_0 + g_1 \tau + \cdots + g_r \tau^r
\]
with $g_i \in K$, $g_r \neq 0$. For a polynomial $h \in K[T]$,
we let $h^\tau$ denote the polynomial
deduced from $h$ by raising all its coefficients to the $q$-th power.
An easy computation then shows that if
$(f_0, \ldots, f_{r-1})$ are the coordinates of some $f \in \M(\phi)$ 
in the canonical basis, then the coordinates $(f'_0, \ldots, f'_{r-1})$ of $\tau f$ are 
 defined by the following matrix equality:
\begin{equation}\label{eq:mult-tau}
  \begin{pmatrix}
    f'_0     \\
    f'_1     \\
    \vdots       \\
    f'_{r-1}
  \end{pmatrix}
  =
  \begin{pmatrix}
    0 & 0 & \dots  & 0 & \frac{T - g_0}{g_r}   \\
    1 & 0 & \dots  & 0 & - \frac{g_1}{g_r}     \\
      &   & \ddots &   &                          \\
    0 & 0 & \dots  & 1 & - \frac{g_{r-1}}{g_r}
  \end{pmatrix}
  \cdot
  \begin{pmatrix}
    f_0^\tau     \\
    f_1^\tau     \\
    \vdots       \\
    f_{r-1}^\tau
  \end{pmatrix}.
\end{equation}
This readily yields Algorithm~\ref{algo:motive-shift-coordinates}.

\begin{algorithm}[h]
    \caption{\MotiveTauAction}
    \label{algo:motive-shift-coordinates}
    \KwIn{The coordinates $(f_0, \dots, f_{r-1})$ of an element $f \in \M(\rho)$}
    \KwOut{The coordinates of $\tau f \in \M(\rho)$}

    Compute the polynomials $f_0^\tau, \dots, f_{r-1}^\tau$ \;
    Compute the polynomial $f'_0 = \frac{T - g_0}{g_r} f^\tau_{r-1}$ \;

    \For{$1 \leqslant i \leqslant r-1$}{
      Compute the polynomial $f'_i = f_i^\tau - \frac{g_{i+1}}{g_r} f^\tau_{r-1}$ \;
    }

    \KwRet $(f'_0, \dots, f'_{r-1})$ \;
\end{algorithm}

\begin{lem}
\label{lemma:comp-motive-shift-coordinates}
  For an input $f \in \M(\phi)$ of $\tau$-degree $n$,
  Algorithm~\ref{algo:motive-shift-coordinates} requires at most $O(n)$
  applications of the Frobenius endomorphism and $O(n)$ operations in $K$.
\end{lem}

\begin{proof}
  By Lemma~\ref{lemma:bound-coeffs}, the polynomial $f_i \in K[T]$ has degree at
  most $\frac {n-i} r$. As a consequence, computing $f_i^\tau$ requires at
  most $\big\lfloor \frac {n-i} r \big\rfloor + 1$ applications of the 
  Frobenius endomorphism, and the pre-computation on line 1 costs 
  \[
    \sum_{i=0}^{r-1} \left(\left\lfloor \frac {n-i} r \right\rfloor + 1\right) = n + 1
  \]
  such applications. The
  remaining steps can be done in $O (n)$ arithmetic operations in~$K$.
\end{proof}

Computing the matrix of $\M(u)$ is now just a matter of computing the
coordinates of $u$ and iteratively applying $r$ times the $\tau$-action.
The precise procedure is presented in Algorithm~\ref{algo:motive-matrix}.

\begin{algorithm}[h]
    \caption{\MotiveMatrix}
    \label{algo:motive-matrix}
    \KwIn{An endomorphism $u : \phi \to \phi$ encoded by its defining Ore polynomial}
    \KwOut{The matrix of $\M(u)$ in the canonical bases}

    Compute $U_0 = \MotiveCoordinates(u, \phi)$ \;

    \For{$1 \leqslant i \leqslant r-1$}{
      Compute $U_i = \MotiveTauAction(U_{i-1})$ \;
    }

    \KwRet the matrix whose columns are $(U_0, \dots, U_{r-1})$ \;

\end{algorithm}

\begin{lem}
\label{lem:complexity-motive-matrix}
  For an input $u$ of $\tau$-degree $n$,
  Algorithm~\ref{algo:motive-matrix} requires 
  at most $O(n^2 + r^2)$ applications of the Frobenius endomorphism, 
  and $O(n^2 + r^2)$ operations in $K$.
\end{lem}

\begin{proof}
  Computing $U_0$ requires $O(n^2)$ applications of the Frobenius endomorphism
  and $O(n^2)$ operations in $K$ (Lemma~\ref{lemma:comp-motive-coordinates}).
  Then, knowing $U_i$ for some $1 \leq i \leq r-1$, the computation of $U_{i+1}$ requires at most
  $O(n+i)$ applications of the Frobenius and $O(n+i)$ operations in~$K$
  by Lemma~\ref{lemma:comp-motive-shift-coordinates}.
  Summing all the contributions, we end up with the announced complexity.
\end{proof}

We now have all the ingredients to write down
Algorithm~\ref{algo:motive-endo-charpoly}, which is the
main algorithm of this section.

\begin{algorithm}[h]
    \caption{\EndomorphismCharpoly}
    \label{algo:motive-endo-charpoly}
    \KwIn{An endomorphism $u : \phi \to \phi$ encoded by its defining Ore polynomial}
    \KwOut{The characteristic polynomial of $u$}
  
    Compute $M = \MotiveMatrix(u)$ \;
    \KwRet the characteristic polynomial of $M$ \;
\end{algorithm}

\begin{theo}\label{theo:motive-charpoly-comp}
  For a morphism of Drinfeld modules $u: \phi \to \phi$ of
  $\tau$-degree~$n$,
  Algorithm~\ref{algo:motive-endo-charpoly} computes the
  characteristic polynomial of $u$ for a cost of
  $O(n^2 + r^2)$ applications of the Frobenius and
  $\Otilde(n^2 + (n{+}r)r^{\Omega - 1})$ operations in $K$.
\end{theo}

\begin{proof}

  The cost of computing the matrix of $\M(u)$ is $O(n^2 + r^2)$
  applications of the Frobenius endomorphism, and $O(n^2 + r^2)$
  operations in $K$. The matrix has size $r$ and, thanks to Corollary
  \ref{cor:bound-coeffs}, we know that all its entries have
  degree less than $1 + \frac n r$.
  Its characteristic polynomial can then be computed with
  $\Otilde((n{+}r) r^{\Omega-1})$ operations in $K$ 
  (see \S \ref{sec:matrices-computations}).
  The theorem follows.
\end{proof}

\subsubsection{The case of finite fields}
\label{sssec:charpoly-finitefield}

If $K$ is a finite field, we can speed up the 
computation by using specific algorithmic primitives to compute 
characteristic polynomial of polynomial matrices (see \S \ref{sec:matrices-computations}) on
the one hand, and to compute Ore Euclidean divisions (see \S 
\ref{subsec:ore-computations}) on the other hand.

\begin{theo}
\label{theo:charpoly-finitefield}
  If $K$ is a finite extension of $\Fq$ of degree $d$ and
  $u$ is an endomorphism of $\tau$-degree~$n$ of a Drinfeld 
  module $\phi$ of rank~$r$, then
  Algorithm~\ref{algo:motive-endo-charpoly} computes the
  characteristic polynomial of $u$ for a cost of
  \[
    \Otilde(d\log^2 q) +
    \Opower\big(\big(\SMgeq(n, d) + ndr + nr^\omega + dr^\omega\big)\cdot\log q\big)
  \]
  bit operations.
\end{theo}

\begin{proof}
  The complexity analysis is similar to that of
  Theorem~\ref{theo:motive-charpoly-comp}, except that the Ore Euclidean division
  of Algorithm~\ref{algo:motive-coordinates} now costs 
  $\Otilde(d \log^2 q) + \Opower(\SMgeq(n, d)\log q)$
  bit operations. The computation of the matrix of $\M(u)$ therefore
  requires 
  \[
    \Otilde(d\log^2 q) +
    \Opower\big((\SMgeq(n, d) + dr(n + r))\log q\big)
  \]
  bit operations.
  Finally, it remains to compute the characteristic polynomial of the matrix. 
  For this, we first notice that all its coefficients of have degree at most $n$ 
  (Corollary~\ref{cor:bound-coeffs}).
  Therefore, using Lemma~\ref{lem:matrice-comp-1}, the computation of the
  characteristic polynomial costs
  $\Opower((n{+}d)r^\omega)$ operations in $\Fq$.
  The theorem follows.
\end{proof}

\begin{rem}
Comparing with the algorithms of~\cite{musleh-schost-2}, we find that 
Algorithm~\ref{algo:motive-endo-charpoly} exhibits a better 
theoretical complexity, except when the degree of $\gamma(T)$ is close
to~$d$ and the rank~$r$ is very small compared to $d$ and $n$; in this
case, the algorithm of~\cite[Theorem 2(1)]{musleh-schost-2} has quadratic 
complexity in $\max(n,d)$, beating the term $\SMgeq(n, d)$.
\end{rem}

When $u$ is the Frobenius endomorphism, 
Algorithm~\ref{algo:motive-endo-charpoly} leads to the algorithm
\FMFF discussed in the introduction, whose complexity is given by
Corollary~\ref{cor:charpoly-finitefield-frobenius}.

\begin{cor}[Variant~\FMFF]
\label{cor:charpoly-finitefield-frobenius}
  If $K$ is a finite field of degree $d$ over $\Fq$,
  Algorithm~\ref{algo:motive-endo-charpoly} computes the
  characteristic polynomial of the Frobenius endomorphism of $\phi$ 
  for a cost of
  \[
    \Otilde(d\log^2 q) +
    \Opower\big(\big(\SMgeq(d, d) + d^2r + dr^\omega\big) \cdot \log q\big)
  \]
  bit operations.
\end{cor}

\begin{proof}
  This is a direct application of Theorem~\ref{theo:charpoly-finitefield} with
  $n = d$.
\end{proof}

\subsubsection{The case of the Frobenius endomorphism: another approach}
\label{sssec:charpoly-kedlaya-umans}

Below, we present yet another method to compute the characteristic polynomial
of the Frobenius endomorphism $F_\phi$.
This leads to the algorithm
\FMKU, as mentioned in the introduction, which performs better for some
ranges of parameters (at least theoretically).
It is based on the following two remarks:
\begin{itemize}[itemsep=0.5ex,parsep=0ex,topsep=0.5ex]
  \item As the Ore polynomial $\tau^d$ is central in $\Ktau$, and its action on the
    motive can unambiguously be defined as a left or right multiplication.
  \item The left multiplication by $\tau$ on $\M(\phi)$ is a semi-linear
    application, whose matrix is the companion matrix appearing 
    in Equation~\eqref{eq:mult-tau}, which is easy to compute.
\end{itemize}
More precisely, for a nonnegative integer $s$, let $\mu_s$ be the $K[T]$-semi-linear
endomorphism of $\M(\phi)$ defined by $f \mapsto \tau^s f$. We denote its
matrix by $M_s$. In other words, $M_s$ is the matrix whose $j$-th column contains
the coefficients of $\tau^{j+s} \in \M(\phi)$ in the canonical basis.
The matrix $M_1$ is the companion matrix of Equation~\eqref{eq:mult-tau}
and, by definition, the matrix of $\M(u)$ is $M_d$.

For a polynomial $P \in K[T]$ and an integer $s$, we define $P^{\tau^s}$ 
as the polynomial obtained by raising each coefficient of $P$ to power $q^s$.
Similarly, given a matrix $M$ with entries in $K[T]$, we write
$M^{\tau^s}$ for the matrix obtained from $M$ by applying $P \mapsto
P^{\tau^s}$ to each of its entry. A calculation shows that
\[
  M_s = M_1 \cdot M_1^\tau \cdots M_1^{\tau^{s-1}}.
\]
This equation leads to the following \emph{square and multiply}-like formulas:
\begin{align}
  M_{2s} & = M_s \cdot M_s^{\tau^s}, \label{eq:Mseven} \\
  M_{2s + 1} & = M_1 \cdot M_s^{\tau} \cdot M_s^{\tau^{s+1}}. \label{eq:Msodd}
\end{align}
Let $\alpha$ be a generator of $K$ over $\Fq$. Elements of $K$ are 
classically represented as polynomials in $\alpha$ with coefficients in 
$\Fq$ and degree $d-1$. Applying $\tau^s$ to an element 
$\sum_{i=0}^{r-1} a_i \alpha^i \in K$ amounts to applying the 
substitution $\alpha \mapsto \tau^s(\alpha)$. Thus, this can be 
efficiently computed using Kedlaya-Umans' algorithm for modular 
composition \cite{kedlaya-umans} for a cost of 
$\Opower(d \log q)$ bit operations. As mentioned in
\S \ref{sssec:complexitymodel}, an initial precomputation of $\alpha^q$ must be
performed once and for all, for a cost of $\Otilde(d \log^2 q)$ bit operations.

\begin{theo}[Variant~\FMKU]
\label{theo:kedlaya-umans}
  If $K$ is a finite extension of $\Fq$ of degree~$d$, the characteristic
  polynomial of the Frobenius endomorphism of a Drinfeld module $\phi$
  of rank~$r$ can be computed for a cost of
  \[
    \Otilde(d\log^2 q) +
    \Opower\big((d^2 r^{\omega - 1} + d r^\omega)\cdot \log q\big)
  \]
  bit operations.
\end{theo}

\begin{proof}
  Let $\mathrm{C}(s)$ be the cost, counted in bit operations, of computing the pair 
  $\mathcal P_s = (M_s, \tau^s(\alpha))$.
  To compute $\mathcal P_{2s}$ and $\mathcal P_{2s+1}$, one uses the
  recurrence relations~\eqref{eq:Mseven} and~\eqref{eq:Msodd}.
  As $M_s$ has $r^2$ polynomial coefficients of degree at most
  $s/r$ (Lemma~\ref{lemma:bound-coeffs}), computing $\tau^s(M)$ requires $O(sr)$ modular compositions of degree
  $d$. As previously mentioned, we use Kedlaya-Umans' algorithm~\cite{kedlaya-umans} for this task, leading to a total cost of
  $\Opower(nrd{\cdot}\log q)$ bit operations. 
  Similarly $\tau^{2s}(\alpha)$ can be computed by composing $\tau^s(\alpha)$
  with itself; using again Kedlaya-Umans' algorithm, this can be done with
  $\Opower(d{\cdot}\log q)$ bit operations. Moreover, the matrix product
  $M_s {\cdot} \tau^s(M)$ requires $\Otilde (dsr^{\omega-1})$ extra operations in~$\Fq$.
  Given that one operation in $\Fq$ corresponds to $\Otilde(\log q) \subset \Opower(\log q)$ bit operations,
  we conclude that
  \[
    \mathrm{C}(2s) \leqslant \mathrm{C}(s) + \Opower(dsr^{\omega - 1}\log q).
  \]
  A similar analysis provides a similar bound for $\mathrm{C}(2s+1)$.
  Solving the recurrence, we obtain
  $\mathrm{C}(s) \in \Opower(d s r^{\omega - 1} \log q)$. 
  Therefore, the computation of $\M(u)$ can be done
  with $\Opower(d^2 r^{\omega - 1} \log q)$ bit operations.

  Finally, the characteristic polynomial of the matrix of $\M(u)$ is computed as
  previously, using Lemma~\ref{lem:matrice-comp-1}, for a cost of $\Otilde(d r^\omega)$ operations in $\Fq$,
  which is no more than $\Opower(d r^\omega{\cdot}\log q)$ bit
  operations. Adding both contributions and taking into account the
  precomputation of $\alpha^q$, we obtain the corollary.
\end{proof}

\subsection{Algorithms: the case of a general curve}

We now drop the assumption that $A = \Fq[T]$.
In full generality, it is not true that the motive $\M(\phi)$
is free over $A_K$, and the matrix of $\M(u)$ is not defined. One can
nevertheless easily work around this difficulty, by extending scalars
to the fraction field of $A_K$, denoted by $\Frac(A_K)$. Indeed, 
$\Frac(A_K) \otimes_{A_K} \M(\phi)$ is obviously free over $\Frac
(A_K)$ given that the latter is a field. It is also clear that 
the determinants of $\M(u)$ and $\Frac(A_K) \otimes_{A_K} \M(u)$ are
equal.

Our first need is to design an algorithm for computing a basis
of $\Frac(A_K) \otimes_{A_K} \M(\phi)$. For this, we will rely on the
case of $\Fq[T]$, previously treated.
We consider an element $T \in A$, $T \not\in \Fq$. Since the underlying 
curve $C$ is absolutely irreducible, $T$ must be transcendental 
over~$\Fq$. This gives an embedding $\Fq[T] \hookrightarrow A$, which 
extends to an inclusion of fields $K(T) \hookrightarrow \Frac(A_K)$. The 
resulting extension is finite of degree $t = \deg(T)$.
Let 
$(b_1, \ldots, b_t)$ be a basis of $\Frac(A_K)$ over $K(T)$.

In what follows, $T$ and $(b_1, 
\ldots, b_t)$ are assumed to be known. Finding them depends on the way 
$C$ is given, but we believe that our hypothesis is reasonable. For instance, if $C$ is presented as a plane smooth
curve, \emph{i.e.} if $A$ is given as
\[
  A = \Fq[X,Y] / P(X,Y)
  \quad \text{with} \quad
  P \in \Fq[X,Y]
\]
one may choose $T = X$, $t = \deg_Y P$ and $b_i = Y^{i-1}$ for $1
\leq j \leq t$.

\begin{rem}
Let $g$ be the genus of $C$.
The Riemann-Roch theorem indicates that the Riemann-Roch space $\mathcal 
L\big((g{+}1){\cdot}[\infty]\big)$ has dimension at least~$2$. Hence it 
must contain a transcendental function, which shows that there always 
exists $T$ for which $t \leq g{+}1$.
In practice, $T$ can be computed through various different algorithms
(see~\cite{legluher-spaenlehauer,abelard-couvreur-lecerf} and the
references therein).
\end{rem}

Now given a Drinfeld module $\phi : A \to K\{\tau\}$ over $A$, we
restrict it to $\Fq[T]$ \emph{via} the
embedding $\Fq[T] \to A$,
obtaining a second Drinfeld module
$\phi' : \Fq[T] \to K\{\tau\}$ (see \S \ref{sssec:restriction}). Then 
$\M(\phi') = \M(\phi)$, with the same structure of $K[T]$-modules.
Moreover, if~$\phi$ has rank~$r$, we have
\[
\deg \phi'_T = \deg \phi_T = r \cdot \deg(T) = rt
\]
showing that $\phi'$ has rank $rt$. The family $(1, \tau,
\ldots, \tau^{rt-1})$ is a basis of $\M(\phi)$ over $K[T]$, and we can
use Algorithm~\ref{algo:motive-coordinates} to compute the coordinates
of any element of $\M(\phi)$ with respect to this basis.
Let $\coord : \M(\phi) \to K[T]^{rt}$ be the map taking an element
of $\M(\phi)$ to the column vector representing its coordinate in
the above basis. Both $\coord$ and $\coord^{-1}$ are efficiently
computable.

Let $e_1$ be an arbitrary nonzero element of $\M(\phi)$, \emph{e.g.}
$e_1 = 1$. A $K(T)$-basis of the $\Frac(A_K)$-line generated by $e_1$
is explicitly given by the family $e_1 \phi_{b_1}, \ldots, e_1 
\phi_{b_t}$. For $1 \leq j \leq t$, we set $C_{1,j} = \coord(e_1
\phi_{b_j})$ and we form the following matrix, with $rt$ rows and $t$
columns:
\[
  M_1 = \left( \begin{matrix}
   C_{1,1} & \cdots & C_{1,t}
  \end{matrix} \right).
\]
We now consider a column vector $E_2$ outside the image of $M_1$
and define $e_2 = \coord^{-1}(E_2)$; $e_2$ is not 
$\Frac(A_K)$-collinear to $e_1$, and we have constructed a free
family of cardinality~$2$. We then continue the same process,  by 
setting $C_{2,j} = \coord(e_2 \phi_{b,j})$ and considering the $rt \times 
2t$ matrix
\[
  M_2 = \left( \begin{matrix}
   C_{1,1} & \cdots & C_{1,t} &
   C_{2,1} & \cdots & C_{2,t}
  \end{matrix} \right).
\]
We pick a column vector $E_3$ outside the image of $M_2$ and
define $e_3 = \coord^{-1}(E_3)$, as well as $M_3$. We repeat this construction until
we reach $e_r$. The vectors $e_1, \ldots, e_r$ being linearly
independent over $\Frac(A_K)$, they form a $\Frac(A_K)$-basis 
of $\Frac(A_K) \otimes_{A_K} \M(\phi)$. The matrix $M_r$ is nothing but the
change-of-basis matrix from the canonical $K[T]$-basis of $\M(\phi)$ to the
newly computed basis 
$\B = (e_1 \phi_{b_1}, \ldots, e_1 \phi_{b_t}, \ldots, 
  e_r \phi_{b_1}, \ldots, e_r \phi_{b_t})$. If $f \in \M(\phi)$, the product 
$M_{r}^{-1} \cdot \coord^{-1}(f)$ gives the coordinates of $f$ in $\B$.
From this, we eventually read the coordinates of $f$ in the
$\Frac(A_K)$-basis $(e_1, \ldots, e_r)$.

To summarize, we have constructed a $\Frac(A_K)$-basis of $\Frac(A_K)
\otimes_{A_K} \M(\phi)$
and designed an algorithm to compute coordinates in this
basis. Using these inputs as primitives, it is now straightforward
to extend the results of \S \ref{subsec:algo-motive} to the case of a general
curve.

\section{Norms of isogenies}
\label{sec:isogenies}

In Section~\ref{sec:endomorphisms}, we have only covered the case of
\emph{endomorphisms} between Drinfeld modules. We now consider
general morphisms and isogenies. Let $\phi, \psi$ be two rank $r$ Drinfeld
$A$-modules, and let $u : \phi \to \psi$ be an isogeny.
In this setting, the caracteristic polynomial is no longer defined but 
the norm of $u$ continues to make sense (see 
\S \ref{ssec:norm}); we recall that it is an ideal of $A$, 
denoted by $\norm(u)$.
The purpose of this section is twofold: first, to establish explicit formulas
that recover $\norm(u)$ at the motive level, and secondly, to offer efficient
algorithms for the computation of $\norm(u)$ using those formulas.

\subsection{Reading norms on the motive}
\label{ssec:normonmotive}

In our general context, the determinant of $u$ can no longer be defined as
previously. In
§\ref{sssec:determinants-proj}, we set up important definitions and
statements about determinants in projective modules. Our main results are
stated in \S \ref{sssec:norm-main-results}.

\subsubsection{Determinants on projective modules}
\label{sssec:determinants-proj}

Let $\A$ be a Dedekind domain. Let $M, M'$ be two finitely generated 
projective $\A$-modules of rank $n$. Let $f: M \to M'$ be
an $\A$-linear mapping. The morphism $f$ gives rise to the $\A$-linear map $\det f :
\det M \to \det M'$. However, when $f$ has different domain and codomain,
\emph{i.e.} $M \neq M'$, it no
longer makes sense to interpret $\det f$ as the multiplication by some
scalar.
Instead, we define the ``determinant'' of $f$, denoted by $\detideal f$,
as the ideal quotient 
$(\det M' : \Im(\det f))$, that is
\[
  \detideal f
  = (\det M' : \Im(\det f))
  = \left\{a \in \A: a \det M' \subset \Im(\det f)\right\}.
\]
Equivalently $\detideal f$ is the annihilator ideal of the cokernel
of $\det f$.

Since $\A$ is a Dedekind domain, $\detideal f$ can be decomposed as a
product
\[
  \detideal f = \prod_{\q} \q^{v_\q(\detideal f)},
\]
where the product runs over all maximal ideals $\q$ of $\A$ and the exponent
$v_\q(\detideal f)$ is a nonnegative integer referred to as the \emph{$\q$-adic
valuation} of $\detideal f$.

For the purpose of this article, it is fundamental to notice that
$v_\q(\detideal f)$ can be found out by computing the classical determinant of an
actual matrix. Indeed, letting as before $\Aq$ denote the
completion\footnote{When studying projective modules, it is more common to
consider the
localization $\A_{(\q)}$ instead of the completion $\A_\q$. Although the
first setting is simpler, the second better suits our needs.} of
$\A$ at $\q$, we define $M_\q = \Aq \otimes_\A M$ and 
$M'_\q = \Aq \otimes_\A M'$. The map $f$ induces a $\Aq$-linear morphism $f_\q : M_\q \to
M'_\q$. We deduce from the flatness of $\Aq$ over $\A$ that
\begin{equation}
\label{eq:deltafq}
  \detideal f_\q = \Aq \otimes_\A \detideal f 
               = (\q{\cdot} \Aq)^{v_\q(\detideal f)},
\end{equation}
where $\detideal f_\q$ is defined, similarly to $\detideal f$, as the
annihilator ideal of the cokernel of $f_\q$.

On the other hand, we know that $\Aq$ is a principal domain. Hence both
$M_\q$ and $M'_\q$ are free of rank $n$ over $\Aq$. We choose bases
$\mathcal B_{(\q)}$ and $\mathcal B'_{(\q)}$ of $M_\q$ and $M'_\q$
respectively, and let $F_\q$ denote the matrix of $f_\q$ in these bases.
It follows from the definition
that $\detideal f_\q = \det(F_\q) \: \Aq$. Comparing with
Equation~\eqref{eq:deltafq}, we finally conclude that
\[
  v_\q(\detideal f) = v_\q(\det F_\q).
\]
We notice in particular that, although the determinant itself depends on the
choices of $\mathcal B_{(\q)}$ and $\mathcal B'_{(\q)}$, its $\q$-adic
valuation does not. Indeed, changing $\mathcal B_{(\q)}$ (resp. $\mathcal
B'_{(\q)}$) boils down to multiplying $F_\q$ by an invertible matrix on the
left (resp. on the right), which only multiplies the determinant a unit, and as
such, does not affect its $\q$-adic valuation.

In a similar fashion, one can relate $\detideal f$ to the Euler-Poincaré
characteristic of the cokernel of~$f$, which is essential to establish
our main theorem.

\begin{prop}\label{prop:coker-delta}
  We have
  \[
    \detideal f = \chi_{\A}(\Coker f).
  \]
\end{prop}

\begin{proof}

  As we have seen, the Euler-Poincaré characteristic commutes with
  localization. Therefore, it is enough to prove that $\detideal f_\q =
  \chi_{\Aq}(\Coker f_\q)$ for each maximal ideal $\q$ of $\A$.

  Let then $\q$ be a maximal ideal of $\A$. It follows from the structure
  theorem of finitely generated modules over principal domains that there exist
  bases $\mathcal B_\q$ and $\mathcal B'_\q$ in which the matrix $F_\q$ of
  $f_\q$ is diagonal. If $\delta_1, \ldots, \delta_r$ denote its diagonal
  coefficients, we have
  \[
    \Coker f_\q \simeq 
    \big(\Aq / \delta_1 \Aq\big) \times \cdots \times
    \big(\Aq / \delta_r \Aq\big).
  \]
  Hence
  \[
    \chi_{\Aq}(\Coker f_\q) 
    = \delta_1 \cdots \delta_r \cdot \Aq 
    = (\det F_\q) {\cdot} \Aq = \detideal f_\q
  \]
  which is what we wanted to prove.
\end{proof}

\subsubsection{Main results}
\label{sssec:norm-main-results}

We may now state and prove the main theoretical results of this subsection.

\begin{theo}\label{th:norm-general}
  Let $\phi$ and $\psi$ be two Drinfeld modules, and let $u: \phi \to \psi$ be
  an isogeny. We have
  \[
    \norm(u) = \detideal \M(u).
  \]
\end{theo}

\begin{proof}
  Writing $u$ as the product of a purely inseparable isogeny
  with a separable isogeny, and noticing that (1)~$\detideal$ is
  multiplicative and (2)~$\M$ is functorial, we are reduced to prove the theorem
  when $u = \tau^{\deg(\p)}$ on the one hand and when $u$ is separable
  on the other hand.

  \textbf{Purely inseparable case.}
  We assume that $u = \tau^{\deg(\p)}$.
  We follow Gekeler's idea for proving the multiplicativity of the norm
  \cite[Lemma~3.10]{gek91}.
  Let $\q \subset \A$ be a maximal ideal away from the characteristic.
  Note that the map $\E_\q(u): \E_\q(\phi) \to \E_\q(\psi)$ is an
  isomorphism because $\tau$ is coprime with the right gcd of
  $\phi_q$ for $q$ varying in $\q$.
  By Theorem~\ref{th:pairing}, we conclude that
  $\M_\q(u) : \M_\q(\psi) \to \M_\q(\phi)$ is an isomorphism as
  well, showing that $\q$ is coprime with $\chi_{A_K}(\Coker \M(u))$.
  Consequently, $\detideal \M(u)$ is a power of~$\p$.
  On the other hand, observe that, by definition,
  \[
    \deg(u) = \dim_K(\Coker \M(u)) = \deg(\chi_{A_K}(\M(u))).
  \]
  Proposition~\ref{prop:coker-delta} then implies that
  $\deg(\detideal \M(u)) = \deg(u) = \deg(\p)$. Putting all together,
  we conclude that $\detideal \M(u) = \p = \norm(u)$.

  \textbf{Separable case.}

  Given that $u$ is nonzero, the kernel of the $A$-linear map $\E(u)$ is a torsion 
  $A$-module. Let $a \in A$ such that $a \cdot \ker \E(u) = 0$.
  For all elements $z \in \Kbar$, we then have the following
  implication: if $u(z) = 0$, then $\phi_a(z) = 0$. Since $u$ is
  separable, this implies that $u$ right-divides $\phi_a$, from
  which we deduce that $a$ annihilates $\Coker \M(u)$ as well.
  Applying successively the right exact functor 
  $- \otimes_A A/aA$ and the left exact functor $\Hom_K(-, \Kbar)$
  to the exact sequence of $A_K$-modules
  \[
    0 \; \to \; \M(\psi) \; \to \; \M(\phi) \; \to \; \Coker \M(u) \; \to \; 0,
  \]
  we get the following exact sequence of $A_{\Kbar}$-modules
  \[
    0
    \; \to \; \big(\Coker\M(u)\big)^\ast \otimes_K \Kbar
    \; \to \; \M_a(\phi)^\ast \otimes_K \Kbar
    \; \to \; \M_a(\psi)^\ast \otimes_K \Kbar.
  \]
  This shows that 
  \[
    \big(\Coker\M(u)\big)^\ast \otimes_K \Kbar 
    \simeq \ker \big(\M_a(u)^\ast\big) \otimes_K \Kbar
    \simeq \ker \big(\M_a(u)^\ast \otimes_K \Kbar\big).
  \]
  From Theorem~\ref{th:pairing}, we then derive the following isomorphisms of
  $A_{\Kbar}$-modules:
  \begin{align*}
    \big(\Coker\M(u)\big)^\ast \otimes_K \Kbar 
    & \simeq \Ker\big(\E_a(u) \otimes_\Fq \Kbar\big) \\
    & = \Ker\big(\E(u) \otimes_\Fq \Kbar\big) \\
    & \simeq \Ker \E(u) \otimes_\Fq \Kbar.
  \end{align*}
  Consequently, $u$ being separable, we find that
  \[ 
    \norm(u) = \chi_A\big(\Ker\E(u)\big) = \chi_{A_K}\big((\Coker \M(u))^\ast\big).
  \]
  Using finally Lemma~\ref{lem:isomorphisme-dual}, we end up with
  $\norm(u) = \chi_{A_K}(\Coker \M(u)) = \detideal \M(u)$, proving the theorem.
\end{proof}

An interesting consequence of Theorem~\ref{th:norm-general} is 
a compatibility result between norms of isogenies and restrictions
of Drinfeld modules (see \S \ref{sssec:restriction}), which will be
particularly useful to us when Drinfeld $A$-modules are restricted 
to $A' = \Fq[T]$.

\begin{cor}
\label{cor:restrictionnorm}
Let $\gamma' : A' \to K$ be a second base for Drinfeld modules 
satisfying the assumptions of \S \ref{ssec:backgrounddrinfeld},
coming together with an injective homomorphism of rings $f : A' \to A$ 
such that $\gamma = \gamma' \circ f$. 
Let $\phi, \psi : A \to K\{\tau\}$ be two Drinfeld $A$-modules and
let $u : \phi \to \psi$ be a morphism. Then
\[
  \norm\big(f^* u\big) = N_{A/A'}\big(\norm(u)\big)
\]
where $N_{A/A'} : A \to A'$ is the norm map from $A$ to $A'$
\emph{via} $f$.
\end{cor}

\begin{proof}
Let $\p$ be a prime ideal of $A'_K$, and let $A'_{K,\p}$ be the
completion of $A'_K$ at $\p$.
Write $A_{K,\p} = A'_{K,\p} \otimes_{A'_K} A_K$,
$\M(\phi)_\p = A'_{K,\p} \otimes_{A'_K} \M(\phi)$, and
$\M(\psi)_\p = A'_{K,\p} \otimes_{A'_K} \M(\psi)$. 
Since $A_{K,\p}$ is a product
of local rings, the module $\M(\phi)_\p$ is free over $A_{K,\p}$.
We pick a basis
 $\B_\phi = (e_{\phi,i})_{1 \leq i \leq r}$ of it,
together with a basis $\B = (a_m)_{1 \leq m \leq n}$ of $A_{K,\p}$ 
over $A'_{K,\p}$.
Note that the family $\B'_\phi = (a_m{\cdot}e_{\phi,i})_{1 \leq i 
\leq r, 1 \leq m \leq n}$ is a $A'_{K,\p}$-basis of $\M(\phi)_\p
= \M(f^* \phi)_\p$.
We define similarly $\B_\psi$ and $\B'_\psi$.
Let $C = (c_{ij})_{1 \leq i,j \leq r}$ be the matrix of $\M(u)$
with respect to the bases $\B_\psi$ and $\B_\phi$ and, for $a
\in A'_{K,\p}$, let $M(a) \in (A'_{K,\p})^{n \times n}$ be the matrix of 
the multiplication by $a$ over $A_{K,\p}$. The matrix of $f^* u$ in the
bases $\B'_\psi$ and $\B'_\phi$ is the block matrix
\[
  D = 
  \left( \begin{matrix}
  M(c_{1,1}) & \cdots & M(c_{1,r}) \\
  \vdots & & \vdots \\
  M(c_{r,1}) & \cdots & M(c_{r,r}) \\
  \end{matrix} \right)
\]
The main result of~\cite{silvester} implies that $\det D = 
N_{A_{K,\p}/A'_{K,\p}}(\det C)$.
The proposition then follows from Theorem~\ref{th:norm-general}.
\end{proof}

\subsection{Algorithms: the case of $\mathbb P^1$}
\label{subsec:norm-p1}

Let $A = \Fq[T]$ as in \S \ref{subsec:algo-motive}.
Theorem~\ref{th:norm-general} readily translates to an algorithm 
for computing the norm of an isogeny between Drinfeld modules;
this is Algorithm~\ref{algo:motive-iso-norm}.

\begin{algorithm}[h]
    \caption{\IsogenyNorm}
    \label{algo:motive-iso-norm}
    \KwIn{An isogeny $u : \phi \to \psi$ encoded by its defining Ore polynomial}
    \KwOut{The norm of $u$}
  
    Compute $M = \MotiveMatrix(u)$ \;
    \KwRet the ideal generated by determinant of $M$
\end{algorithm}

\begin{theo}
\label{theo:iso-norm}
  Let $\phi$ and $\psi$ be two Drinfeld $\Fq[T]$-modules of rank $r$
  and let $u: \phi \to \psi$ be an isogeny of $\tau$-degree $n$.
  Algorithm~\ref{algo:motive-iso-norm} computes the norm of~$u$ for
  a cost of $O (n^2 + r^2)$ applications of the Frobenius endomorphism of $K$
  and $\Otilde(n^2 + nr^{\omega-1} + r^\omega)$ operations in~$K$.

\end{theo}

\begin{proof}
  Per Lemma~\ref{lem:complexity-motive-matrix}, the cost of computing the
  matrix of $\M(u)$ is $O(n^2{+}r^2)$ applications of the Frobenius
  endomorphism, and $O(n^2{+}r^2)$ operations in $K$. Besides,
  this matrix has size $r$ and its entries have degrees all less than 
  $1 + \frac n r$ (Lemma~\ref{lemma:bound-coeffs}, which is also
  valid for isogenies).
  Therefore, using the algorithmic primitives of
  $\S \ref{sec:matrices-computations}$, computing its determinant requires $\Otilde((n{+}r) r^{\omega-1})$
  operations in~$K$.
\end{proof}

When $K$ is a finite field, one can speed up
Algorithm~\ref{algo:motive-iso-norm} using the optimized primitives
of \S \ref{subsec:ore-computations} for manipulating Ore polynomials, as for
the endomorphism case.
Precisely, we have the following.

\begin{theo}
\label{theo:iso-norm-finite}
  If $K$ is a finite field of degree $d$ over $\Fq$,
  Algorithm~\ref{algo:motive-iso-norm} computes the norm of the isogeny $u$ for
  a cost of
  \[
    \Otilde(d\log^2 q) +
    \Opower\big(\big(\SMgeq(n, d) + ndr + n\min(d,r)r^{\omega-1} + dr^\omega\big)
    \cdot \log q\big)
  \]
  bit operations.
\end{theo}

\begin{proof}
  Per the first part of the proof of Theorem~\ref{theo:charpoly-finitefield},
  the computation of $\M(u)$ requires 
  \[
    \Otilde(d \log^2 q) + \Opower\big((\SMgeq(n, d) + dr(n+r))\log q\big)
  \]
  bit operations.
  Then, for the computation of the determinant, we distinguish
  between two cases. If $d \leq r$, we keep on using the algorithms
  of~\cite{matrices:poly-computations-03, matrices:poly-computations-05}, for a
  cost of $\Otilde(nr^{\omega-1} + r^\omega)$ operations in~$K$, that is
  $\Otilde(ndr^{\omega-1} + dr^\omega)$ operations in~$\Fq$. On the contrary,
  when $d \geq r$, we use Lemma~\ref{lem:matrice-comp-1}, performing then
  $\Opower(nr^\omega + dr^\omega)$ operations in~$\Fq$. Putting all together,
  and remembering that an operation in~$\Fq$ corresponds to $\Otilde(\log q)$
  bit operations, we get the theorem.
\end{proof}

\begin{rem}

  When $u$ is an endomorphism, the norm can be computed as the constant
  coefficient of the characteristic polynomial of $u$, up to a sign. We notice
  that the algorithms of the present subsection in some cases run faster than
  those of \S \ref{subsec:algo-motive}. This is because we compute the
  determinant of the matrix of $\M(u)$ instead of its whole characteristic
  polynomial. However, we stress that the asymptotic costs of computing the
  characteristic polynomial and the norm of an endomorphism may be equal. This
  owes to the fact that in some cases, computing the characteristic polynomial
  of a matrix, or computing its determinant, both reduces to matrix
  multiplication.

\end{rem}

\begin{rem}
\label{rem:frobenius-norm}
In the special case where $u = F_\phi$ is the Frobenius endomorphism, the norm is given
by a simple closed-formula (see \cite[Theorem~2.11]{gek08} and
\cite[Theorem~2.4.7]{papikian_drinfeld_2023}), namely
\begin{equation}
\label{eq:frobenius-norm}
  \norm(F_\phi) = (-1)^{rd - r - d} N_{K/\Fq}(\Delta)^{-1}
  \p^{\frac{d}{\deg(\p)}},
\end{equation}
where $\Delta$ is the leading coefficient of $\phi_T$. Computing 
the Frobenius norm using Equation~\eqref{eq:frobenius-norm} costs 
$\Otilde(d \log^2 q) + \Opower(d \log q)$ bit 
operations~\cite[Proposition~3]{musleh-schost-1}. Noticing that the
\emph{Frobenius norm} is a degree $d$ polynomial in $\Fq[T]$, this complexity
is essentially optimal with respect to~$d$, and asymptotically better 
than other algorithms mentioned in this paper (see also 
Appendix~\ref{appendix:review}).
\end{rem}

\subsection{Algorithms: the case of a general curve}
\label{ssec:normcurve}

When $A$ is arbitrary, determining the norm of an isogeny $u: \phi \to \psi$
becomes more complex due to the nonfreeness of the motives $\M(\phi)$ and
$\M(\psi)$ in general. This necessitates working with arbitrary torsion-free
modules over Dedekind rings. While this approach appears viable, we will 
follow an alternative strategy that simplifies the general scenario by
reducing the computation to the previously addressed case of $\Fq[T]$.

From now on, we assume for simplicity that $A$ is presented as
\[
  A = \Fq[X,Y] / P(X,Y)
\]
and that $\deg(x) > \deg(y)$, where $x$ and $y$ denote the images 
in $A$ of $X$ and $Y$ respectively.
Let $\phi, \psi : A \to K\{\tau\}$ be two Drinfeld modules of rank~$r$, 
and let $u : \phi \to \psi$ be an isogeny between them. We consider a 
new variable $\Lambda$ and form the polynomial rings $K[\Lambda]$ and 
$A_K[\Lambda]$. We set
\[
  \M(\phi)[\Lambda] = A_K[\Lambda] \otimes_{A_K} \M(\phi)
\]
and endow it with the structure of $K[T,\Lambda]$-module 
inherited from its 
structure of $A_K[\Lambda]$-module through the ring homomorphism
\[
  f : K[T, \Lambda] \;\to\; A_K[\Lambda], \quad 
  T \mapsto x + \Lambda{\cdot}y, \,\,
  \Lambda \mapsto \Lambda.
\]
Similarly, we define $\M(\psi)[\Lambda]$ and endow it with a structure
of $K[T,\Lambda]$-module.

The assumption $\deg(x) > \deg(y)$ ensures that $\phi_x + \Lambda
{\cdot}\phi_y$ is an Ore polynomial of degree $r {\cdot} \deg(x)$ 
with leading coefficient lying in~$K$. Writing $s = r {\cdot} \deg(x)$, 
we deduce that the family $(1, \tau, \ldots, \tau^{s-1})$ is a 
$K[T,\Lambda]$-basis of both $\M(\phi)[\Lambda]$ and $\M(\psi)
[\Lambda]$. 
On the other hand, we observe that, after extending scalars to 
$A_K[\Lambda]$, the morphism $\M(u) : \M(\psi) \to \M(\phi)$ induces a 
$K[T,\Lambda]$-linear map $\M(u)[\Lambda] : 
\M(\psi)[\Lambda] \to \M(\phi)[\Lambda]$.
Its determinant in the aforementioned distinguished bases is a 
bivariate polynomial, that we call $\delta(T,\Lambda)$. Evaluating it at 
$T = x + \Lambda y$, we obtain a univariate polynomial in $\Lambda$
with coefficients in $A_K$.

\begin{theo}
\label{theo:normcurve}
With the above notation and hypothesis, the leading coefficient of
$\delta(T,\Lambda)$ with respect to $T$ is a nonzero constant $c \in
K^\times$.
Moreover, if we write
\[
  \delta(x{+}\Lambda y, \Lambda) =
  \delta_0 + \delta_1{\cdot}\Lambda + \cdots + \delta_n{\cdot}\Lambda^n
  \qquad (n \in \NN, \delta_i \in A_K),
\]
then $c^{-1} \delta_0, \ldots, c^{-1} \delta_n$ all lie in $A$ and generate
$\norm(u)$.
\end{theo}

\begin{proof}
For any fixed element $\lambda \in \Kbar$, notice that the degree of
the univariate polynomial $\delta(T, \lambda)$ is equal to the $\tau$-degree
of~$u$. Since the 
latter remains constant when $\lambda$ varies in $\Kbar$, so does the
former. The first assertion of the theorem follows.

Set $I = \Kbar \otimes_{Fq} \norm(u)$, which is an ideal of $A_{\Kbar}$. 
Recall that the maximal ideals of $A_{\Kbar}$ are all of the form
\[
  \m_{(x_0, y_0)} = (x-x_0) A_{\Kbar} + (y-y_0) A_{\Kbar}
\]
with $x_0, y_0 \in \Kbar$. We write the decomposition of $I$ into a
product of prime ideals:
\begin{equation}
\label{eq:decompI}
  I = \m_{(x_1, y_1)} \cdot \m_{(x_2, y_2)} \cdots \m_{(x_\ell,y_\ell)}
\end{equation}
where $\ell$ is a nonnegative integer and $x_i, y_i \in \Kbar$ for
all $i$ between $1$ and $\ell$.

We fix an element $\lambda \in \Kbar$ and consider the ring 
homomorphism $f_\lambda : \Kbar[T] \to A_{\Kbar}$ defined by $T \mapsto x +
\lambda y$. The map $f_\lambda$ is the specialization of $f$ at $\lambda$, and a finite 
morphism whose degree does not depend on $\lambda$.
Let $N_\lambda : A_{\Kbar} \to \Kbar[T]$ denote the norm map with respect to 
$f_\lambda$. It follows from the decomposition~\eqref{eq:decompI}
that $N_\lambda(I)$
is the ideal of $\Kbar[T]$ generated by the polynomial
\[
  P_\lambda(T) = 
  (T - x_1 - \lambda y_1) \cdots (T - x_\ell - \lambda y_\ell).
\]
On the other hand, repeating the proof of
Corollary~\ref{cor:restrictionnorm}, we find that $N_\lambda(I)$
is also the ideal
generated by $\delta(T, \lambda)$. Therefore $\delta(T, \lambda) =
c \cdot P_\lambda(T)$. Since this equality holds for any $\lambda
\in \Kbar$, it is safe to replace $\lambda$ by the formal variable
$\Lambda$. Specializing at $T = x + \Lambda y$, we obtain
\[
  \delta(x + \Lambda y, \Lambda) = 
  c \cdot \prod_{i=1}^\ell \big((x - x_i) + \Lambda{\cdot}(y - y_i)\big)
\]
Expanding the latter product and comparing with the definition 
of~$I$, we find that $I$ is the ideal of $A_{\Kbar}$ generated by
$\delta_0, \ldots, \delta_n$.
Finally, the fact that $I$ is defined over $A$ implies that the
pairs $(x_i, y_i)$ are conjugated under the Galois action, which
eventually shows that the $c^{-1} {\cdot} \delta_i$'s are in $A$.
The theorem follows.
\end{proof}

Theorem~\ref{theo:normcurve} readily translates to an algorithm for
computing the norm $\norm(u)$, namely:
\begin{enumerate}[itemsep=0.5ex,parsep=0ex,topsep=0.5ex]
\item we compute the matrix of $\M(u)[\Lambda]$ using 
Algorithm~\ref{algo:motive-matrix}
(treating $\Lambda$ as a formal parameter),
\item we compute the determinant $\delta(T,\Lambda)$ of this matrix
and let $c \in K^\times$ be its leading coefficient with respect to
$T$,
\item we write
\[
  c^{-1} \cdot \delta(x{+}\Lambda y, \Lambda) =
  \delta'_0 + \delta'_1{\cdot}\Lambda + \cdots + \delta'_n{\cdot}\Lambda^n
  \qquad (\delta'_i \in A_K).
\]
\item we return the ideal of $A$ generated by $\delta'_0, \ldots, \delta'_n$.
\end{enumerate}

\medskip

It follows from the proof of Theorem~\ref{theo:normcurve} that the degree 
$n$ of $\delta(x{+}\Lambda y, \Lambda)$ is equal to $\ell$, on the one 
hand, and to the $\tau$-degree of the isogeny~$u$, on the other hand. 
Unfortunately, this quantity may be large, especially when we compare 
it with the minimal number of generators of $\norm(u)$, which is at most 
$2$ because $A$ is a Dedeking domain.

To overcome this issue, an option could be to compute the 
$\delta'_i$'s one by one by using relaxed arithmetics~\cite{vdH97}: each
time a new $\delta'_i$ is computed, we form the ideal $I_i$ generated
by $\delta'_0, \ldots, \delta'_i$ and stop the process when $I_i$
has degree~$n$; we then have the guarantee that $\norm(u) = I_i$ and that
we have computed the ideal we were looking for.
When $x_1, \ldots, x_\ell$ are pairwise disjoint (which is the
most favorable case), we already have $\norm(u) = I_1$, so that the 
above procedure stops very rapidly.

Another option consists in picking random elements $\lambda \in K$ and 
computing the evaluations $\delta(T, \lambda)$ and $c^{-1} {\cdot} 
\delta(x{+}\lambda y, y)$. Doing so, we obtain elements in $\norm(u)$ and 
we can hope, as above, that only a few number of them will generate the 
ideal. Again, this can be checked by looking at the degree of the 
candidate ideals.

\section{The \emph{central simple algebra} method}
\label{sec:CSA}

Throughout this section, we assume that $K$ is a finite extension 
of $\Fq$ and we let $d$ denote the degree of $K/\Fq$.
Our aim is to design an alternative algorithm (namely the algorithm
referred to as \FCSA in the introduction) for computing the
characteristic polynomial of the Frobenius endomorphism $F_\phi$
of a rank $r$ Drinfeld $A$-module $\phi$.
We recall that, by definition, $F_\phi$ is the endomorphism 
corresponding to the Ore polynomial $\tau^d \in \Ktau$.

Our algorithm is based on Theorem~\ref{theo:charpolyNrd}, which
provides a formula for the characteristic polynomial of
$F_\phi$ by means of reduced norms in a certain central simple
algebra.

\subsection{The characteristic polynomial of the Frobenius as a reduced norm}
\label{ssec:charpolyNrd}

Theorem~\ref{theo:charpolyNrd}, the main result of this section and stated in
\S\ref{sssec:csa-main_results}, requires a preliminary introduction on general
Ore polynomials and reduced norms. This is the goal of
\S\ref{sssec:ore_pols-reduced_norms}.

\subsubsection{General Ore polynomials and reduced norms}
\label{sssec:ore_pols-reduced_norms}

We first recall some standard facts about Ore
polynomials\footnote{For a more detailed survey on this topic,
we refer to~\cite[\S I]{jacobson}.}.
Given a ring $L$ equipped with a ring endomorphism $\varphi :
L \to L$, we form the ring $L[\ttau;\varphi]$ whose elements are formal
expressions of the form
$$a_0 + a_1 \ttau + \cdots + a_n \ttau^n \quad
  (n \in \NN, \, a_0, \ldots, a_n \in L)$$
subject to the usual addition and multiplication driven by the rule
$\ttau b = \varphi(b) \ttau$ for $b \in L$. The ring $L[\ttau;\varphi]$ is the
so-called ring of {Ore polynomials} over $L$ twisted by $\varphi$;
it is noncommutative unless $\varphi$ is the identity morphism.

From this point onward, we focus on the case where $L$ is a field, as it holds
significant importance for this section.
The ring $L[\ttau;\varphi]$ then shares many properties with classical
polynomial rings over a field. Notably, it is equipped with a 
notion of degree and with an Euclidean division on the right: given 
two Ore polynomial $A, B \in L[\ttau;\varphi]$ with $B \neq 0$, there exist
uniquely determined $Q, R \in L[\ttau;\varphi]$ such that $A = QB + R$ and
$\deg R < \deg B$. As in the classical commutative case, this implies
that $L[\ttau;\varphi]$ is left Euclidean, \emph{i.e.} all left ideals of
$L[\ttau;\varphi]$ are generated by one element. From this property, we
derive the existence of right gcd: given $P, Q \in L[\ttau;\varphi]$,
the right gcd of $P$ and $Q$, denoted by $\rgcd(P,Q)$, is the unique
monic polynomial satisfying the relation
\[
  L[\ttau;\varphi]{\cdot}P +  L[\ttau;\varphi]{\cdot}Q = 
  L[\ttau;\varphi]{\cdot}\rgcd(P,Q).
\]

From now on, we assume further that $\varphi$ has finite order~$d$.
This hypothesis ensures in particular that the center of $L[\ttau;\varphi]$
is large; precisely, it is the subring $F[\ttau^d]$ where $F$ denotes the
subfield of $L$ fixed by $\varphi$. By standard Galois theory, the
extension $L/F$ has degree~$d$ and it is Galois with cyclic Galois
group generated by~$\varphi$.
In this situation, the field of fractions of $L[\ttau;\varphi]$ can be
obtained by inverting the elements in the center, \emph{i.e.} we
have
\[
  \Frac(L[\ttau;\varphi]) = F(\ttau^d) \otimes_{F[\ttau^d]} L[\ttau;\varphi].
\]
Besides, the latter is a central simple algebra over 
$F(\ttau^d)$~\cite[Theorem~1.4.6]{jacobson}.
This provides us with a reduced norm map
\[
  \Nrd : \quad \Frac(L[\ttau;\varphi]) \;\to\; F(\ttau^d)
\]
which is multiplicative and acts as the $d$-th power on $F(\ttau^d)$.
Let $P \in L[\ttau;\varphi]$, $P \neq 0$.
We form the quotient $D_P = L[\ttau;\varphi] / L[\ttau;\varphi] P$, which is
a $L$-vector space of dimension $\deg(P)$ with basis $(1, x, \ldots,
\ttau^{\deg(P) -1})$. Since $\ttau^d$ is a central element in $L[\ttau;\varphi]$,
the multiplication by $\ttau^d$ defines a $L$-linear endomorphism of
$D_P$, which we denote by $\gamma_P$. Its characteristic polynomial
$\pi(\gamma_P)$ is then a monic polynomial of degree~$\deg(P)$.

\begin{prop}
\label{prop:charpolyNrd}
For all $P \in L[\ttau;\varphi]$, $P \neq 0$, we have
\[
  \Nrd(P) = N_{L/F}\big(\lc(P)\big) \cdot \pi(\gamma_P)(\ttau^d)
\]
where $\lc(P)$ is the leading coefficient of $P$ and $N_{L/F}$
is the norm map from $L$ to $F$, \emph{i.e.} $N_{L/F}(x) = 
x \cdot \varphi(x) \cdots \varphi^{r-1}(x)$.
\end{prop}

\begin{proof}
See \cite[Lemma~2.1.15]{caruso-leborgne-1}.
\end{proof}

\begin{rem}
Proposition~\ref{prop:charpolyNrd} implies in particular that
$\Nrd(P)$ is a polynomial whenever $P \in L[\ttau;\varphi]$ and that
$\pi(\gamma_P)$ has coefficients in~$F$. Both of them are not
immediate from the definition.
\end{rem}

\subsubsection{Main results}
\label{sssec:csa-main_results}

We come back to our setting: we assume that $K$ is a finite 
extension of $\Fq$ of degree~$d$ and consider a Drinfeld module 
$\phi : A \to K\{\tau\}$ of rank~$r$.
We notice that $K\{\tau\}$ can be alternatively depicted as the
ring of Ore polynomials $K[\ttau; \Frob]$ where $\Frob : K \to K$ is
the Frobenius endomorphism taking $x$ to $x^q$.
Recall that we have set $A_K = K \otimes_{\Fq} A$, and define
$\varphi = \Frob \otimes\: \id_A$, which is a ring endomorphism of
$A_K$ of order~$d$ with fixed subring~$A$. 
We form the Ore algebra $A_K[\ttau; \varphi]$; it contains $K[\ttau;\varphi]
\simeq K\{\tau\}$ as a subring. In particular, the elements $\phi_a$
$(a \in A)$ naturally sit in $A_K[\ttau;\varphi]$.

We define the ideal
\[
  I(\phi) = \sum_{a \in A} A_K[\ttau; \varphi] {\cdot} (\phi_a - a).
\]
In other words, $I(\phi)$ is the left ideal of $A_K[\ttau; \varphi]$
generated by the elements $(\phi_a - a)$ for $a$ running over $A$.

\begin{lem}
\label{lem:Iphi}
We assume that $A$ is generated as a $\Fq$-algebra by the elements
$a_1, \ldots, a_n$. Then $I(\phi)$ is generated as a left ideal
of $A_K[\ttau; \varphi]$ by $\phi_{a_1}{-}a_1, \, \ldots, \,
\phi_{a_n}{-}a_n$.
\end{lem}

\begin{proof}
Let $I'$ be the left ideal of $A_K[\ttau; \varphi]$ generated by $\phi_{a_1}{-}a_1,
\, \ldots, \, \phi_{a_n}{-}a_n$. We need to prove that $I' = I(\phi)$.
The inclusion $I' \subset I(\phi)$ is obvious.
For the reverse inclusion, consider $\lambda \in \Fq$ and $a,b \in A$ 
such that $\phi_a{-}a, \, \phi_b{-}b \in I'$. The equalities
\begin{align*}
\phi_{\lambda a} - \lambda a 
 & = \lambda \cdot (\phi_a - a) \\
\phi_{a+b} - (a{+}b) 
 & = (\phi_a - a) + (\phi_b - b) \\
\phi_{ab} - ab 
 & = \phi_a \cdot (\phi_b - b) + b \cdot (\phi_a - a)
\end{align*}
(recall that $b$ is central, so it commutes with $\phi_a$)
show that the three elements on the left hand side belong to $I'$
as well. This stability property eventually ensures that $I'$
contains all elements of the form $\phi_a - a$. Hence $I(\phi)
\subset I'$ as desired.
\end{proof}

We recall from \S \ref{ssec:motive} that the $A$-motive of $\phi$, 
denoted by $\M(\phi)$, is isomorphic to $K\{\tau\}$ as a $K$-vector 
space. This gives
a $K$-linear inclusion $\M(\phi) \hookrightarrow A_K[\ttau; \varphi]$
(mapping $\tau$ to $\ttau$). We consider the composite
\[
  \alpha_\phi : \quad \M(\phi) \; \hookrightarrow \; A_K[\ttau; \varphi]
  \; \to \; A_K[\ttau; \varphi] / I(\phi).
\]

\begin{prop}
\label{prop:motiveOre}
The map $\alpha_\phi$ is a $A_K$-linear isomorphism.
\end{prop}

\begin{proof}
We first check linearity.
Let $\lambda \in K$, $a \in A$ and $f \in \M(\phi)$. By definition,
we have $(\lambda \otimes a) {\cdot} f = \lambda f \phi_a$. Hence
\[
  \alpha_\phi\big((\lambda \otimes a) {\cdot} f\big)
= \lambda f \phi_a \equiv \lambda f a \pmod{I(\phi)}.
\]
Moreover $a$ is a central element in $A_K[\ttau; \varphi]$.
We conclude that 
$\alpha_\phi\big((\lambda \otimes a) {\cdot} f\big)
= \lambda a f$ and linearity follows.

In order to prove that $\alpha_\phi$ is an isomorphism, we 
observe that $A_K[\ttau; \varphi] \simeq K\{\tau\} \otimes_{\Fq} A$
and we define the $K$-linear map
$\beta_\phi : A_K[\ttau; \varphi] \to K\{\tau\}$ (as sets, $\Ktau = \M(\phi)$)
that takes $f \otimes a$ to $f \phi_a$ (for $f \in K\{\tau\}$
and $a \in A$).
We claim that $\beta_\phi$ vanishes on $I(\phi)$. Indeed, 
for $a, b \in A$ and $g \in K\{\tau\}$, we have
\begin{align*}
    \beta_\phi\big((g \otimes b){\cdot}(\phi_a \otimes 1 - 1 \otimes a)\big)
& = \beta_\phi(g\phi_a \otimes b - g \otimes ab) \\
& = g\phi_a \phi_b - g \phi_{ab} = 0.
\end{align*}
Consequently, $\beta_\phi$ induces a mapping
$\bar \beta_\phi : A_K[\ttau; \varphi]/I(\phi) \to \M(\phi)$.
It is now formal to check that $\bar \beta_\phi$ is a left and 
right inverse of $\alpha_\phi$, showing that $\alpha_\phi$ is an
isomorphism.
\end{proof}

We write $\Frac(A_K)$ for the field of fractions of $A_K$.
The morphism $\varphi$ extends to a ring endomorphism of $\Frac(A_K)$
that, in a slight abuse of notation, we continue to denote by 
$\varphi$. On $\Frac(A_K)$, $\varphi$ has order~$d$ and its fixed
subfield is $\Frac(A)$. 
We consider the Ore polynomial ring $\Frac(A_K)[\ttau; \varphi]$. By what we 
have seen previously, its center is $\Frac(A)[\ttau^d]$
and there is a reduced norm map
\functionname
  {\Nrd}
  {\Frac(A_K)[\ttau; \varphi]}
  {\Frac(A)[\ttau^d].}
We define $I_0(\phi) = 
\Frac(A_K) \otimes_{A_K} I(\phi)$; it is a left ideal 
of~$\Frac(A_K)[\ttau; \varphi]$. Since the latter is a principal ideal domain, $I_0(\phi)$
is generated by a unique element $g(\phi)$, which we assume to be
monic.
Concretely $g(\phi)$ is the right gcd of the elements $(\phi_a - a)$
when $a$ varies in $A$. After Lemma~\ref{lem:Iphi}, we even have
$g(\phi) = 
 \rgcd\big(\phi_{a_1}{-}a_1, \, \ldots, \phi_{a_n}{-}a_n\big)$
as soon as $a_1, \ldots, a_n$ generate $A$ as an $\Fq$-algebra.

\begin{theo}
\label{theo:charpolyNrd}
We keep the previous notation and assumptions.
Let $F_\phi$ be the Frobenius endomorphism of $\phi$ and 
let $\pi(F_\phi)$ be its monic characteristic polynomial. Then
\[
  \pi(F_\phi)(\ttau^d) = \Nrd\big(g(\phi)\big).
\]
\end{theo}

\begin{proof}
Write $\M_0(\phi) = \Frac(A_K) \otimes_{A_K} \M(\phi)$.
On the one hand, it follows from Proposition~\ref{prop:motiveOre}
that $\alpha_\phi$ induces an isomorphism
\[
  \M_0(\phi) \simeq \Frac(A_K)[\ttau; \varphi]\,/\,\Frac(A_K)[\ttau; \varphi]{\cdot}g(\phi).
\]
With Proposition~\ref{prop:charpolyNrd}, we realize that
$\Nrd(g(\phi))$ is equal to the characteristic polynomial of
the right multiplication by $\ttau^d$ on $\M_0(\phi)$, that is
\[
  \Nrd\big(g(\phi)\big) 
= \pi\big(\Frac(A_K) \otimes_{A_K} \M(F_\phi)\big)
= \pi\big(\M(F_\phi)\big).
\]
We conclude by invoking Theorem~\ref{th:norm-endomorphism}.
\end{proof}

\begin{rem}\label{rem:csa}

  When $A = \Fq[T]$, we recover a result given in \cite{papikian_drinfeld_2023}
  (see Lemma~4.3.1, Theorems~4.2.2 and 1.7.16, and Equation (4.1.3)). It
  follows from Lemma~\ref{lem:Iphi} that $g(\phi)$ is $K(T)$-collinear to
  $\phi_T - T$. Therefore, the reduced norm of $\phi_T - T$ corresponds to the
  reduced characteristic polynomial of $\phi_T$. Let $\chi(\tau^d, V) \in
  \Fq[\tau^d][V]$ be this characteristic polynomial, and let $\pi$ be the
  characteristic polynomial of the Frobenius endomorphism of $\phi$. Then we
  have shown the polynomials $\pi(T, X)$ and $\chi(X, T)$ are equal up to a
  nonzero element in $\Fq$.

\end{rem}

\subsection{Algorithms: the case of $\mathbb P^1$}
\label{ssec:CSAP1}

We move to algorithmical purpose.
By Theorem~\ref{theo:charpolyNrd}, the computation of the characteristic 
polynomial of $F_\phi$ reduces to the computation of a reduced norm.
On the other hand, it is a classical fact that the 
reduced norm of a polynomial $P \in A_K[\ttau;\varphi]$ can be computed as 
a usual norm. Precisely, we consider the subalgebra $A[\ttau]$ of 
$A_K[\ttau;\varphi]$; it is commutative. 
Moreover $A_K[\ttau;\varphi]$ appears as a free left module of rank $d$ over 
$A[\ttau]$. Thus, there exists a norm map
$N_{A_K[\ttau;\varphi]/A[\ttau]}$ which takes a polynomial $P$ to the 
determinant of the $A[\ttau]$-linear endomorphism of
\functiondefname
  {\mu_P}
  {A_K[\ttau;\varphi]}
  {A_K[\ttau;\varphi]}
  {Q}
  {QP.} 
With this notation, we have
\[
  \Nrd(P) = N_{A_K[\ttau;\varphi]/A[\ttau]}(P) \in A[\ttau].
\]

We now assume that $A = \Fq[T]$ and
fix a Drinfeld module $\phi : \Fq[T] \to K\{\tau\}$. It follows
from Lemma~\ref{lem:Iphi} that $g(\phi)$ is $K(T)$-collinear 
to $\phi_T - T$.
Fix a basis $\mathcal B = (e_1, \dots, e_d)$ of $K$ over
$\Fq$ and observe that $\mathcal B$ is an $A[\ttau]$-basis of 
$A_K[\ttau;\varphi]$ as well.
Let $M$ be the matrix of $\mu_{\phi_T}$ in $\mathcal B$. Its entries all 
lie in $\Fq[\ttau]$ given that $\phi_T$ has coefficients in $K$. Observing
moreover that $\mu_{g(\phi)} = \mu_{\phi_T} - \mu_T = \mu_{\phi_T} - T$,
we conclude that 
\begin{equation}
\label{eq:charpolyFphiFqT}
  \pi(F_\phi)(\ttau^d) = \pi(M)(T)
\end{equation}
where $\pi(M)$ is the characteristic polynomial of $M$.
We emphasize that the two variables $t$ and $T$ play different
roles in the two sides of the Equality~\eqref{eq:charpolyFphiFqT}:
in the left hand side, $t$ appears in the variable at which the 
characteristic polynomial is evaluated whereas, in the right hand
side, it is an internal variable appearing in the matrix $M$; and
conversely for~$T$.

In order to explicitly compute the matrix of $\mu_P$ for a 
given Ore polynomial $P \in K[\ttau;\varphi]$, we can proceed as
follows. We write 
$P = g_0 + g_1 t + \cdots + g_n \ttau^n$ ($g_i \in K$) and
notice that
\[
  \mu_P = 
  \mu_{g_0} + \mu_\ttau \circ \mu_{g_1} + \cdots + \mu_\ttau^n \circ \mu_{g_n}.
\]
Moreover the set of equalities $e_i \ttau = \ttau e_i^{1/q}$ for $1 \leq i 
\leq d$ shows that the matrix 
of $\mu_\ttau$ is $\ttau{\cdot}F^{-1}$ where $F$ is the matrix of the Frobenius
endomorphism acting on~$K$ (which is $\Fq$-linear).
These observations readily lead to Algorithm~\ref{algo:ore-matrix}.

\begin{algorithm}[h]
    \caption{\CSAMatrix}
    \label{algo:ore-matrix}
    
    \KwIn{An Ore polynomial $P = \sum_{j=0}^n g_j \ttau^j \in K[\ttau;\varphi]$,
      a basis $\mathcal B = (e_1, \ldots, e_d)$ of $K$ over~$\Fq$}
    \KwOut{The matrix of $\mu_P$ in the basis $\mathcal B$}

    Compute the matrix $F \in \FF_q^{d \times d}$ of the Frobenius
    $K \to K, x \mapsto x^q$ in the basis~$\mathcal B$\;

    \For { $0 \leq j \leq n$}{
        Compute the matrix $G_j \in \FF_q^{d \times d}$
        of the map $K \to K, x \mapsto g_j x$ in the basis~$\mathcal B$\
    }

    \KwRet $\sum_{j=0}^n F^{-j} {\cdot} G_j {\cdot} \ttau^j$ \;
\end{algorithm}

\begin{lem}
\label{lem:complexité-matrice-csa}
  If $\mathcal B$ is the working basis of $K/\Fq$,
  Algorithm~\ref{algo:ore-matrix} requires $d$
  applications of the Frobenius endomorphism and $\Otilde(n d^\omega)$ 
  operations in $\Fq$.
\end{lem}

\begin{proof}
  Since $\mathcal B$ is the working basis, writing the coordinates of
  an element of~$K$ in $\mathcal B$ costs nothing. Therefore,
  computing the matrix $F$ amounts to computing each $g_i^q$ for $1 \leq i \leq
  d$. This then requires $d$ applications of the
  Frobenius endomorphism. Similarly computing each $G_j$ requires $d$ 
  multiplications in~$K$, corresponding to $\Otilde(d^2)$ operations
  in~$\Fq$.
  Finally, the computation on line 4 requires one inversion and $O(n)$ 
  multiplications of $r \times r$ matrices over~$\Fq$. The cost of this 
  computation is then $\Otilde(n d^\omega)$ operations in $\Fq$.
\end{proof}

We now have everything we need to compute the characteristic polynomial
of the Frobenius endomorphism: see Algorithm~\ref{algo:frobenius-charpoly}.

\begin{algorithm}[h]
    \caption{\FrobeniusCharpoly}
    \label{algo:frobenius-charpoly}
    
    \KwIn{A Drinfeld $\Fq[T]$-module $\phi$}
    \KwOut{The characteristic polynomial of the Frobenius endomorphism of
    $\phi$}

    Compute $M = \CSAMatrix(\phi_T)$ \;
    Compute the characteristic polynomial of $M$ and write it
    $\sum_{i=0}^d (\sum_{j=0}^{r} \lambda_{i, j} \ttau^{jd}) X^i$ \;

    \KwRet $\sum_{j = 0}^r (\sum_{i=0}^d \lambda_{i, j} T^j) X^i$ \;
\end{algorithm}

\begin{theo}[Variant~\FCSA]
\label{theo:CSA}
  Algorithm~\ref{algo:frobenius-charpoly} computes
  the characteristic polynomial of the Frobenius endomorphism of $\phi$ for a
  cost of $\Otilde(d\log^2 q) + \Opower(r d^\omega \log q)$ bit operations.
\end{theo}

\begin{proof}
  Per Lemma~\ref{lem:complexité-matrice-csa}, computing the matrix of
  $\mu_{\phi_T}$ requires $O(d)$ applications of the Frobenius and
  $\Otilde(rd^\omega)$ operations in $\Fq$.
  Using Lemma~\ref{lem:matrice-comp-2}, computing its characteristic 
  polynomial can be achieved for an extra cost of $\Otilde(rd^\omega)$
  operations in $\Fq$.
  All of this correspond to $\Otilde(d\log^2 q) + \Opower(r d^\omega \log q)$
  bit operations in our complexity model (see \S \ref{sssec:complexitymodel}).
\end{proof}

\subsection{Algorithms: the case of a general curve}

When $A$ is a general curve, it is possible to follow the same strategy as
before. However several simplifications that were previously applicable cannot
be implemented in this case.
First of all, finding $g(\phi)$ requires some computation.
By Lemma~\ref{lem:Iphi}, however, $g(\phi)$ can be obtained as the right
gcd of a finite number of Ore polynomials, as soon as we have a
finite presentation of the ring $A$.
Fortunately, such a right gcd can be computed using a noncommutative
variant of the Euclidean algorithm.
Once $g(\phi)$ is known, one can compute its reduced norm using the 
method of \S \ref{ssec:CSAP1}: we form the matrix of the 
$\Frac(A)[\ttau]$-linear map $\mu_{g(\phi)} : \Frac(A_K)[\ttau;\varphi] \to 
\Frac(A_K)[\ttau;\varphi]$, defined by $Q \mapsto Q {\cdot} g(\phi)$, and view 
$\Nrd(g(\phi))$ as the determinant of $\mu_{g(\phi)}$.

This approach yields a working algorithm for computing $\pi(F_\phi)$.
It has nevertheless two drawbacks.
First, the computation of the right gcd may be costly and have an
impact on the size of the coefficients in the base ring $\Frac(A_K)$,
which is not finite. One may gain a certain level of control
by using the theory of noncommutative subresultants introduced 
by Li in~\cite{li98}, but this requires additional caution.
The second disadvantage is that the Ore polynomial $g(\phi)$ is
in general not of the form $\phi_a - a$, implying that the
computation of its reduced norm no longer boils down to finding
the characteristic polynomial of a matrix with entries in~$\Fq$.
Instead, we need to compute the determinant of a general matrix
over $\Frac(A)[\ttau]$, which can be a more costly operation.

It turns out that we can overcome these two issues by following
the same strategy as in \S \ref{ssec:normcurve} and reducing the
problem to the case of $\Fq[T]$.
For simplicity, we assume again that $A$ is presented as
\[
  A = \Fq[X,Y] / P(X,Y)
  \quad \text{with} \quad
  P \in \Fq[X,Y]
\]
and that $\deg(x) > \deg(y)$ where $x$ and $y$ denote the images
in $A$ of the variables $X$ and $Y$. We introduce a new variable
$\Lambda$ and the Ore polynomial ring $K[T,\Lambda][\ttau; \varphi]$
where $\varphi$ acts on $K$ \emph{via} the Frobenius map $x \mapsto x^q$
and acts trivially on $T$ and $\Lambda$. In this setting, we have
a reduced norm map
\[
  \Nrd : K[T,\Lambda][\ttau; \varphi] \to \Fq[T,\Lambda][\ttau^d].
\]
We consider the trivariate polynomial
$\varpi(T, \Lambda, \ttau^d) = \Nrd\big(\phi_x + \Lambda{\cdot}\phi_y
- T\big)$ and write
\[
  \varpi(x+\Lambda y, \Lambda, \ttau^d) =
  \varpi_0(\ttau^d) + \varpi_1(\ttau^d) {\cdot} \Lambda + \cdots +
  \varpi_n(\ttau^d) {\cdot} \Lambda^n
\]
where the $\varpi_i$'s are univariate polynomials over $\Frac(A)$.
This gives the following theorem, which is an analogue of
Theorem~\ref{theo:normcurve} and whose proof is similar.

\begin{theo}
\label{theo:charpolycurve}
We keep the previous notation and assumptions.
Let $F_\phi$ be the Frobenius endomorphism of $\phi$ and 
let $\pi(F_\phi)$ be its monic characteristic polynomial. Then
\[
  \pi(F_\phi) = \gcd(\varpi_0, \varpi_1, \ldots, \varpi_n).
\]
\end{theo}

The formula of Theorem~\ref{theo:charpolycurve} readily provides an algorithm
for computing $\pi(F_\phi)$. This strategy is not hindered by the two
aforementioned disadvantages. Moreover, as mentioned in \S
\ref{ssec:normcurve}, it may occur that $\pi(F_\phi)$ is already the gcd of the
first polynomials $\varpi_0, \ldots, \varpi_i$, for some $i < n$. Therefore, it
can be beneficial to compute the $\varpi_i$'s one by one (using relaxed
arithmetics), determining the corresponding gcd at each step, and stopping the
computation as soon as the resulting polynomial reaches degree~$d$. As also
discussed in \S \ref{ssec:normcurve}, another option is to work with
evaluations at random values $\lambda \in \Kbar$ instead of working with the
formal variable $\Lambda$.

\bibliographystyle{alpha} 
\bibliography{norm}

\begin{thebibliography}{ACLM23}

\bibitem[ACL22]{abelard-couvreur-lecerf}
Simon Abelard, Alain Couvreur, and Grégoire Lecerf.
\newblock Efficient computation of {Riemann}–{Roch} spaces for plane curves
  with ordinary singularities.
\newblock {\em AAECC}, 2022.

\bibitem[ACLM23]{software-presentation}
David Ayotte, Xavier Caruso, Antoine Leudière, and Joseph Musleh.
\newblock Drinfeld modules in {SageMath}.
\newblock {\em ACM Communications in Computer Algebra}, 57(2):65--71, 2023.

\bibitem[And86]{anderson_t-motives_1986}
Greg~W. Anderson.
\newblock t-{Motives}.
\newblock {\em Duke Mathematical Journal}, 53(2):457--502, June 1986.
\newblock Publisher: Duke University Press.

\bibitem[Ang94]{angles}
Bruno Anglès.
\newblock {\em Modules de Drinfeld sur les corps finis}.
\newblock PhD thesis, 1994.

\bibitem[BCDA22]{bombar_codes_2022}
Maxime Bombar, Alain Couvreur, and Thomas Debris-Alazard.
\newblock On {Codes} and {Learning} with {Errors} over {Function}
  {Fields}.
\newblock In {\em Advances in {Cryptology} – {CRYPTO} 2022}, Lecture {Notes}
  in {Computer} {Science}, pages 513--540. Springer Nature Switzerland, 2022.

\bibitem[Car35]{carlitz_certain_1935}
Leonard Carlitz.
\newblock On certain functions connected with polynomials in a {Galois} field.
\newblock {\em Duke Mathematical Journal}, 1(2), June 1935.

\bibitem[Car18]{caranay-thesis}
Perlas Caranay.
\newblock {\em Computing isogeny volcanoes of rank two {D}rinfeld Modules}.
\newblock PhD thesis, University of Calgary, 2018.

\bibitem[CG24]{caruso-gazda}
Xavier Caruso and Quentin Gazda.
\newblock Computation of classical and $v$-adic $l$-series of $t$-motives.
\newblock {\em RNT}, to appear, 2024.

\bibitem[CGS20]{brenner_computing_2020}
Perlas Caranay, Matthew Greenberg, and Renate Scheidler.
\newblock Computing modular polynomials and isogenies of rank two {Drinfeld}
  modules over finite fields.
\newblock {\em Contemporary mathematics}, 754:283--313, 2020.

\bibitem[Che40]{chevalley_theorie_1940}
Claude Chevalley.
\newblock La {Th\'eorie} du {Corps} de {Classes}.
\newblock {\em Annals of Mathematics}, 41(2):394--418, 1940.

\bibitem[CL09]{couveignes-lercier}
Jean-Marc Couveignes and Reynald Lercier.
\newblock Elliptic periods for finite fields.
\newblock {\em Finite Fields Appl.}, 15(1):1--22, 2009.

\bibitem[CL13]{couveignes-lercier-2}
Jean-Marc Couveignes and Reynald Lercier.
\newblock Fast construction of irreducible polynomials over finite fields.
\newblock {\em Israel J. Math.}, 194(1):77--105, 2013.

\bibitem[CLB17a]{caruso-leborgne-2}
Xavier Caruso and Jérémy Le~Borgne.
\newblock Fast multiplication for skew polynomials.
\newblock {\em Proceedings of the 2017 International Symposium on Symbolic and
  Algebraic Computation}, 2017.

\bibitem[CLB17b]{caruso-leborgne-1}
Xavier Caruso and Jérémy Le~Borgne.
\newblock A new faster algorithm for factoring skew polynomials over finite
  fields.
\newblock {\em Journal of Symbolic Computation}, 79:411--443, 2017.

\bibitem[CLRS22]{cormen_introduction_2022}
Thomas~H. Cormen, Charles~E. Leiserson, Ronald~L. Rivest, and Clifford Stein.
\newblock {\em Introduction to {Algorithms}, fourth edition}.
\newblock MIT Press, 2022.

\bibitem[Con09]{conrad_history_2009}
Keith Conrad.
\newblock History of {Class} {Field} {Theory}.
\newblock 2009.

\bibitem[DNS21]{doliskani-narayanan-schost}
Javad Doliskani, Anand~Kumar Narayana, and Éric Schost.
\newblock Drinfeld modules with complex multiplication, hasse invariants and
  factoring polynomials over finite fields.
\newblock {\em Journal of Symbolic Computation}, 105:199--213, 2021.

\bibitem[Dri74]{drinfeld-paper}
Vladimir~G. Drinfeld.
\newblock Elliptic modules.
\newblock {\em Mathematics of the Ussr-Sbornik}, 23(4):561--592, 1974.

\bibitem[DWZ22]{matrices:omega}
Ran Duan, Hongxun Wu, and Renfei Zhou.
\newblock {Faster matrix multiplication via asymmetric hashing}, 2022.
\newblock Technical Report 2210.10173, arXiv.

\bibitem[Eis95]{eisenbud}
David Eisenbud.
\newblock {\em Commutative Algebra with a View Toward Algebraic Geometry}.
\newblock Springer, 1995.

\bibitem[Gek91]{gek91}
Ernst-Ulrich Gekeler.
\newblock On finite {D}rinfeld modules.
\newblock {\em Journal of algebra}, 1(141):187--203, 1991.

\bibitem[Gek08]{gek08}
Ernst-Ulrich Gekeler.
\newblock Frobenius distributions of drinfeld modules over finite fields.
\newblock {\em Transactions of the American Mathematical Society},
  4(360):1695--1721, 2008.

\bibitem[GJV03]{matrices:poly-computations-03}
Pascal Giorgi, Claude-Pierre Jeannerod, and Gilles Villard.
\newblock On the complexity of polynomial matrix computations.
\newblock In {\em Proceedings of the 2003 International Symposium on Symbolic
  and Algebraic Computation}, ISSAC '03, pages 135--142. Association for
  Computing Machinery, 2003.

\bibitem[GL20]{grishkov_introduction_2020}
Alexandre Grishkov and Dmitry Logachev.
\newblock Introduction to {Anderson} t-motives: a survey, August 2020.
\newblock arXiv:2008.10657v3.

\bibitem[Gos98]{gos98}
David Goss.
\newblock {\em Basic Structures of Function Field Arithmetic}.
\newblock Springer, 1998.

\bibitem[GP20]{garai-papikian}
Sumita Garai and Mihran Papikian.
\newblock Endomorphism rings of reductions of {D}rinfeld modules.
\newblock {\em Journal of Number Theory}, 212:18--39, 2020.

\bibitem[Hay74]{hayes_explicit_1974}
David~R. Hayes.
\newblock Explicit class field theory for rational function fields.
\newblock {\em Transactions of the American Mathematical Society},
  189(0):77--91, 1974.

\bibitem[Hay11]{hayes_brief_2011}
David~R. Hayes.
\newblock A {Brief} {Introduction} to {Drinfeld} {Modules}.
\newblock In {\em A {Brief} {Introduction} to {Drinfeld} {Modules}}, pages
  1--32. De Gruyter, June 2011.

\bibitem[Hil32]{hilbert_neuer_1932}
David Hilbert.
\newblock Ein neuer {Beweis} des {Kroneckerschen} {Fundamentalsatzes} über
  {Abelsche} {Zahlkörper}.
\newblock In David Hilbert, editor, {\em Gesammelte {Abhandlungen}: {Erster}
  {Band} {Zahlentheorie}}, pages 53--62. Springer, 1932.

\bibitem[Jac96]{jacobson}
Nathan Jacobson.
\newblock {\em Finite-dimensional division algebras over fields}.
\newblock Springer-Verlag, Berlin, 1996.

\bibitem[JN19]{joux_drinfeld_2019}
Antoine Joux and Anand~Kumar Narayanan.
\newblock Drinfeld modules may not be for isogeny based cryptography, 2019.
\newblock Report Number: 1329.

\bibitem[JV05]{matrices:poly-computations-05}
Claude-Pierre Jeannerod and Gilles Villard.
\newblock Asymptotically fast polynomial matrix algorithms for multivariable
  systems.
\newblock {\em International Journal of Control}, 79:1359--1367, 2005.

\bibitem[Kal92]{matrices:kaltofen}
Erich Kaltofen.
\newblock On computing determinants of matrices without divisions.
\newblock In {\em Proceedings of the 1992 International Symposium on Symbolic
  and Algebraic Computation}, ISSAC '92. Association for Computing Machinery,
  1992.

\bibitem[Kat73]{katz}
Nicholas~M. Katz.
\newblock {$p$}-adic properties of modular schemes and modular forms.
\newblock In {\em Modular functions of one variable, {III} ({P}roc. {I}nternat.
  {S}ummer {S}chool, {U}niv. {A}ntwerp, {A}ntwerp, 1972)}, pages 69--190.
  Lecture Notes in Mathematics, Vol. 350. Springer, 1973.

\bibitem[Kro53]{kronecker-weber}
Leopold Kronecker.
\newblock Über die algebraisch auflösbaren gleichungen.
\newblock In K.~Hensel, editor, {\em Leopold Kronecker’s Werke, Part 4},
  pages 4--11. American Mathematical Society, 1853.

\bibitem[KU11]{kedlaya-umans}
Kiran~S. Kedlaya and Christopher Umans.
\newblock Fast {Polynomial} {Factorization} and {Modular} {Composition}.
\newblock {\em SIAM Journal on Computing}, 40(6):1767--1802, January 2011.
\newblock Publisher: Society for Industrial and Applied Mathematics.

\bibitem[KV05]{matrices:kaltofen-villard}
Erich Kaltofen and Gilles Villard.
\newblock On the complexity of computing determinants.
\newblock {\em Computational Complexity}, 13(3–4):91--130, feb 2005.

\bibitem[Laf02]{lafforgue}
Laurent Lafforgue.
\newblock Chtoucas de {D}rinfeld, formule des traces d'{A}rthur-{S}elberg et
  correspondance de {L}anglands.
\newblock In {\em Proceedings of the {I}nternational {C}ongress of
  {M}athematicians, {V}ol. {I} ({B}eijing, 2002)}, pages 383--400. Higher Ed.
  Press, Beijing, 2002.

\bibitem[Lau95]{laumon_cohomology_1995}
Gérard Laumon.
\newblock {\em Cohomology of Drinfeld Modular Varieties}.
\newblock Cambridge Studies in Advanced Mathematics. Cambridge University
  Press, 1995.

\bibitem[LGS20]{legluher-spaenlehauer}
Aude Le~Gluher and Pierre-Jean Spaenlehauer.
\newblock A fast randomized geometric algorithm for computing {R}iemann-{R}och
  spaces.
\newblock {\em Math. Comp.}, 89(325):2399--2433, 2020.

\bibitem[Li98]{li98}
Ziming Li.
\newblock A subresultant theory for {O}re polynomials with applications.
\newblock In {\em Proceedings of the 1998 {I}nternational {S}ymposium on
  {S}ymbolic and {A}lgebraic {C}omputation}, pages 132--139. ACM, New York,
  1998.

\bibitem[LS24]{leudiere_computing_2024}
Antoine Leudière and Pierre-Jean Spaenlehauer.
\newblock Computing a group action from the class field theory of imaginary
  hyperelliptic function fields.
\newblock {\em Journal of Symbolic Computation}, 125, 2024.

\bibitem[MS19]{musleh-schost-1}
Yossef Musleh and \'Eric Schost.
\newblock Computing the characteristic polynomial of a finite rank two
  {Drinfeld} module.
\newblock In {\em Proceedings of the 2019 International Symposium on Symbolic
  and Algebraic Computation}, ISSAC '19, pages 307--314. Association for
  Computing Machinery, 2019.

\bibitem[MS23]{musleh-schost-2}
Yossef Musleh and Éric Schost.
\newblock Computing the characteristic polynomial of endomorphisms of a finite
  {Drinfeld} module using crystalline cohomology.
\newblock In {\em Proceedings of the 2023 International Symposium on Symbolic
  and Algebraic Computation}, {ISSAC} '23, pages 461--469. Association for
  Computing Machinery, 2023.

\bibitem[Nar18]{narayanan}
Anand~Kumar Narayanan.
\newblock Polynomial factorization over finite fields by computing
  {Euler}–{Poincaré} characteristics of {D}rinfeld modules.
\newblock {\em Finite Fields and Their Applications}, 54:335--365, 2018.

\bibitem[NP21]{matrices:field-charpoly-mult}
Vincent Neiger and Clément Pernet.
\newblock Deterministic computation of the characteristic polynomial in the
  time of matrix multiplication.
\newblock {\em Journal of Complexity}, 67:101572, 2021.

\bibitem[Pap23]{papikian_drinfeld_2023}
Mihran Papikian.
\newblock {\em Drinfeld {Modules}}, volume 296 of {\em Graduate {Texts} in
  {Mathematics}}.
\newblock Springer International Publishing, 2023.

\bibitem[Poo22]{poonen_introduction_2022}
Bjorn Poonen.
\newblock Introduction to {Drinfeld} modules.
\newblock {\em Arithmetic, Geometry, Cryptography, and Coding Theory}, 779,
  January 2022.

\bibitem[PS07]{matrices:field-charpoly}
Clément. Pernet and Arne Storjohann.
\newblock Faster algorithms for the characteristic polynomial.
\newblock In {\em Proceedings of the 2007 International Symposium on Symbolic
  and Algebraic Computation}, ISSAC '07, pages 307--314. Association for
  Computing Machinery, 2007.

\bibitem[Ros02]{rosen_number_2002}
Michael Rosen.
\newblock {\em Number {Theory} in {Function} {Fields}}, volume 210 of {\em
  Graduate {Texts} in {Mathematics}}.
\newblock Springer, 2002.

\bibitem[Sca01]{scanlon_public_2001}
Thomas Scanlon.
\newblock Public {Key} {Cryptosystems} {Based} on {Drinfeld} {Modules} {Are}
  {Insecure}.
\newblock {\em Journal of Cryptology}, 14(4):225--230, September 2001.

\bibitem[Sil00]{silvester}
John~R. Silvester.
\newblock Determinants of block matrices.
\newblock {\em The Mathematical Gazette}, 84(501):460--467, 2000.

\bibitem[Tak14]{takagi_collected_2014}
Teiji Takagi.
\newblock {\em Collected {Papers}}.
\newblock Springer {Collected} {Works} in {Mathematics}. Springer Tokyo, 2
  edition, November 2014.

\bibitem[vdH97]{vdH97}
Joris van~der Hoeven.
\newblock Lazy multiplication of formal power series.
\newblock In {\em Proceedings of the 1997 International Symposium on Symbolic
  and Algebraic Computation}, ISSAC '97, pages 17--20. Association for
  Computing Machinery, 1997.

\bibitem[vdH04]{weil-pairing}
Gert~Jan van~der Heiden.
\newblock Weil pairing for {D}rinfeld modules.
\newblock {\em Monatshefte f{\"u}r Mathematik}, 143:115--143, 2004.

\bibitem[VS06]{villa_salvador_topics_2006}
Gabriel~Daniel Villa~Salvador.
\newblock {\em Topics in the {Theory} of {Algebraic} {Function} {Fields}}.
\newblock Mathematics: {Theory} \& {Applications}. Birkhäuser, 2006.

\bibitem[vzGG13]{gathen}
Joachim von~zur Gathen and J\"{u}rgen Gerhard.
\newblock {\em Modern computer algebra}.
\newblock Cambridge University Press, Cambridge, third edition, 2013.

\end{thebibliography}

\newpage

\appendix
\section{Review of existing algorithms}
\label{appendix:review}

In all this Section, $\phi$ is a rank $r$ Drinfeld $\Fq[T]$-module over a field
$K$. The field $K$ may not be finite, but when it is, its degree over $\Fq$ is
denoted by $d$. The function field characteristic of $K$ is an ideal $\p$ of
$\Fq[T]$ whose degree is denoted by $m$. We consider an endomorphism or an
isogeny $u$ whose degree as an Ore polynomial is $n$. 
We let $\omega$ be a feasible exponent for matrix multiplication and 
$\Omega$ be a feasible exponent for matrix characteristic polynomial 
computation.

We underline that any algorithm the computes the characteristic 
polynomial of an endomorphism computes its norm as a byproduct.
Furthermore, the Frobenius norm can be computed in
$\Otilde(d \log^2 q) + \Opower(d \log q)$ bit operations (see
Remark~\ref{rem:frobenius-norm}), which is strictly better than 
any other algorithm mentioned in this paper.

In all the tables below, the term $\Otilde(d \log^2 q)$ which appears
in blue on many lines always correspond to the precompution of the image 
of a generator of $K/\Fq$ by the Frobenius endomorphism (see \S 
\ref{sssec:complexitymodel}).

\newcommand{\precomp}{{\color{blue} {} + \Otilde(d \log^2 q)}}

\begin{center}
\medskip

  \begin{threeparttable}

  \captionsetup{font=footnotesize}
  \caption{Algorithms for the characteristic polynomial of the Frobenius
  endomorphism in rank two\hfill\null}

  \footnotesize
  \begin{tabular}{|p{2.7cm}|p{8.2cm}|p{1.4cm}|}
    \hline
    \multicolumn{1}{|c|}{\textit{Algorithm}}
    & \multicolumn{1}{c|}{\textit{Bit complexity}}
    & \multicolumn{1}{c|}{\textit{Constraints}} \\
    \hline
    \hline

    \cite{gek08}
    \tnote{1}
    & $\Opower(d^3 \log q) \precomp$
    & \cellcolor{gray!15} \\ \hline

    \cite[\S~5]{musleh-schost-1}
    \tnote{2}
    & $\Opower(d^{1.885}\log q) \precomp$
    & $m = d$ \\ \hline

    \cite[\S~7]{musleh-schost-1}
    \tnote{3}
    & $\Opower(d^2\log^2 q)$
    & \cellcolor{gray!15} \\ \hline

    \cite[\S~6]{musleh-schost-1}
    \tnote{4}
    & $\Opower(d^2 \log q) \precomp$
    & \cellcolor{gray!15} \\ \hline

    \cite[\S~5.1]{garai-papikian}
    \tnote{\sh}
    & $\Otilde (d^3 \log q)$
    & $m = d$ \\ \hline

    \cite[Th~1]{doliskani-narayanan-schost} 
    \tnote{5}
    & $\Opower(d^{1.5} \log q)
         \precomp$
    & $m = d$ \\ \hline

    \cite[Th.~1(1)]{musleh-schost-2}
    \tnote{\sh}
    & $\Opower(d^{1.5} \log q)
         \precomp$
    & $m = d$ \\ \hline

    \cite[Th.~1(2)]{musleh-schost-2}
    \tnote{\sh}
    & $\Opower\big(\frac {d^2} {\sqrt{m}} \log q\big)
        \precomp$
    & $m < d$ \\ \hline

    \cite[Th.~2(1)]{musleh-schost-2}
    \tnote{\fl}
    & $\Opower(d^2 \frac {d + m}{m} \log q)
        \precomp$
    & \cellcolor{gray!15} \\ \hline

    \cite[Th.~2(2)]{musleh-schost-2}
    \tnote{\fl}
    & $\Opower(\SMgeq(d, d) \log q)
        \precomp$
    & \cellcolor{gray!15} \\ \hline

    \textbf{Cor.~\ref{cor:charpoly-finitefield-frobenius}, \FMFF}
    \tnote{\sh}
    & $\Opower(\SMgeq(d, d) \log q)
        \precomp$
    & \cellcolor{gray!15} \\ \hline

    \textbf{Th.~\ref{theo:kedlaya-umans}, \FMKU}
    \tnote{\sh}
    & $\Opower(d^2 \log q)
        \precomp$
    & \cellcolor{gray!15} \\ \hline

    \textbf{Th.~\ref{theo:CSA}, \FCSA}
    \tnote{\sh}
    & $\Opower(d^\omega \log q)
        \precomp$
    & \cellcolor{gray!15} \\ \hline

  \end{tabular}

  \begin{tablenotes}

    \footnotesize

    \item [1] Deterministic algorithm by Gekeler. The Frobenius norm is
      directly computed, and the Frobenius trace is computed as the solution of
      a linear system. See also \cite[\S~4.1]{musleh-schost-1}.

    \item [2] Monte-Carlo algorithm by Musleh and Schost. The algorithm is
      inspired by ideas from ideas of Narayanan in \cite[\S~3.1]{narayanan}, as
      well as Copersmith's block Wiedemann algorithm.

    \item [3] Monte-Carlo algorithm by Musleh and Schost. The algorithm
      computes the Frobenius norm, and the minimal polynomial of $\phi_T$ using
      a Monte-Carlo algorithm. After, it recovers $F_\phi$ by solving a Hankel
      system.

    \item [4] Deterministic algorithm by Musleh and Schost. Drinfeld analogue
      of Schoof's algorithm for elliptic curves.

    \item [5] Deterministic Algorithm by Doliskani,
      Narayanan and Schost, introduced to factorize polynomials in $\Fq[T]$.
      The algorithm actually computes the \emph{Hasse invariant} of the
      Drinfeld module, from which the Frobenius trace is recovered thanks to
      the assumption that $m = d$. The algorithm gets inspiration from
      elliptic curve algorithms and computes the Hasse invariant as an element
      in a recursive sequence discovered by Gekeler. See
      \cite[\S~2.1]{doliskani-narayanan-schost}.

    \item [\sh] Algorithm described in Table~\ref{table:second}.

    \item [\fl] Algorithm described in Table~\ref{table:third}.

  \end{tablenotes}

  \end{threeparttable}

\vspace{1cm}

  \begin{threeparttable}

  \captionsetup{font=footnotesize}
  \caption{Algorithms for the characteristic polynomial of the Frobenius
  endomorphism in any rank $r$\hfill\null}
  \label{table:second}

  \footnotesize
  \begin{tabular}{|p{2.7cm}|p{8.2cm}|p{1.4cm}|}
    \hline
    \multicolumn{1}{|c|}{\textit{Algorithm}}
    & \multicolumn{1}{c|}{\textit{Bit complexity}}
    & \multicolumn{1}{c|}{\textit{Constraints}} \\
    \hline
    \hline

    \cite[\S~5.1]{garai-papikian}
    \tnote{1}
    & $\Otilde (r^2 d^3 \log q)$
    & $m = d$ \\ \hline

    \cite[Th.~1(1)]{musleh-schost-2}
    \tnote{2}
    & $\Opower(r^\omega d^{\frac 3 2} \log q) \precomp$
    & $m = d$ \\ \hline

    \cite[Th.~1(2)]{musleh-schost-2}
    \tnote{2}
    & $\Opower\big(\big(\frac {r^\Omega} m + \frac {r^\omega} {\sqrt{m}}\big) d^2\log q\big)
        \precomp$
    & $m < d$ \\ \hline

    \cite[Th.~2(1)]{musleh-schost-2}
    \tnote{\fl}
    & $\Opower\big(\big(r^\Omega + \min(dr^2, (d{+}r)r^{\omega-1})\big) \frac {d(d + m)}{m} \log q\big)
        \precomp$
    & \cellcolor{gray!15} \\ \hline

    \cite[Th.~2(2)]{musleh-schost-2}
    \tnote{\fl}
    & $\Opower\big(\big(r^\Omega \frac{d(d + m)}{m} + 
                   r {\cdot} \SMgeq(d + r, d)\big) \log q\big)
        \precomp$
    & \cellcolor{gray!15} \\ \hline

    \textbf{Cor.~\ref{cor:charpoly-finitefield-frobenius}, \FMFF}
    \tnote{3}
    & $\Opower((\SMgeq(d, d) + rd^2 + dr^\omega) \log q)
        \precomp$
    & \cellcolor{gray!15} \\ \hline

    \textbf{Th.~\ref{theo:kedlaya-umans}, \FMKU}
    \tnote{4}
    & $\Opower((d^2 r^{\omega-1} + dr^\omega) \log q)
        \precomp$
    & \cellcolor{gray!15} \\ \hline

    \textbf{Th.~\ref{theo:CSA}, \FCSA}
    \tnote{5}
    & $\Opower(rd^\omega \log q)
        \precomp$
    & \cellcolor{gray!15} \\

    \hline

  \end{tabular}

  \begin{tablenotes}

    \footnotesize

    \item [1] Deterministic algorithm by Garai and Papikian. With
      Proposition~\ref{prop:borne-coeff-frob-charpoly} and the hypothesis $m = d$, the
      coefficients of $F_\phi$ are uniquely determined by their images under
      $\gamma: \Fq[T] \to K$. The Frobenius norm is computed using
      Equation~\eqref{eq:frobenius-norm} and the other coefficients are
      recursively computed.

    \item [2] Two deterministic algorithms by Musleh and Schost. The
      characteristic polynomial of any endomorphism is the characteristic
      polynomial of its action on the crystaline cohomology. In the case of the
      Frobenius endomorphism, algorithmic speed-ups are possible using a
      \emph{baby step-giant step} method.

    \item [3] Probabilistic algorithm. The characteristic polynomial of the
      Frobenius endomorphism is the characteristic polynomial of its action on
      the motive.

    \item [4] Probabilistic algorithm. The characteristic polynomial of the
      Frobenius endomorphism is the characteristic polynomial of its action on
      the motive. The corresponding matrix is recursively
      computed using a \emph{square and multiply}-like procedure.

    \item [5] Probabilistic algorithm. The characteristic polynomial of the
      Frobenius endomorphism is interpreted as the reduced characteristic
      polynomial of $\phi_T$ in the central simple $\Fq[\tau^d]$-algebra $\Ktau$.

    \item [\fl] Algorithm described in Table~\ref{table:third}.

  \end{tablenotes}

  \end{threeparttable}

\vspace{1cm}

  \begin{threeparttable}

  \captionsetup{font=footnotesize}
  \caption{Algorithms for characteristic polynomials of degree $n$
  endomorphisms, in any rank $r$, over a finite field of degree $d$ over
  $\Fq$\hfill\null}
  \label{table:third}

  \footnotesize
  \begin{tabular}{|p{2.7cm}|p{8.2cm}|p{1.4cm}|}
    \hline
    \multicolumn{1}{|c|}{\textit{Algorithm}}
    & \multicolumn{1}{c|}{\textit{Bit complexity}}
    & \multicolumn{1}{c|}{\textit{Constraints}} \\
    \hline
    \hline

    \cite[Th.~2(1)]{musleh-schost-2}
    \tnote{1}
    & $\Opower\big(\big(r^\Omega + \min(nr^2, (n{+}r)r^{\omega-1})\big) \frac {d(n + m)}{m} \log q\big)
        \precomp$
    & \cellcolor{gray!15} \\ \hline

    \cite[Th.~2(2)]{musleh-schost-2}
    \tnote{1}
    & $\Opower\big(\big(r^\Omega \frac{d(n + m)}{m} + r \SMgeq(n + r, d)\big) \log q\big)
        \precomp$
    & \cellcolor{gray!15} \\ \hline

    \textbf{Th.~\ref{theo:charpoly-finitefield}, \FMFF}
    \tnote{2}
    & $\Opower((\SMgeq(n, d) + ndr + nr^\omega + dr^\omega) \log q)
        \precomp$
    & \cellcolor{gray!15} \\ \hline

  \end{tabular}

  \begin{tablenotes}
    \footnotesize

    \item [1] Two deterministic algorithms by Musleh and Schost. The
      characteristic polynomial of any endomorphism is the characteristic
      polynomial of its action on the crystalline cohomology of the Drinfeld
      module.

    \item [2] Probabilistic algorithm. The characteristic polynomial of any
      endomorphism is the characteristic polynomial of its action on the motive
      of the Drinfeld module.

  \end{tablenotes}

  \end{threeparttable}

\vspace{1cm}

  \begin{threeparttable}

  \captionsetup{font=footnotesize}
  \caption{Algorithms for characteristic polynomials of degree $n$
  endomorphisms, in any rank $r$, over a generic field\hfill\null}
  \label{table:fourth}

  \footnotesize
  \begin{tabular}{|p{2.7cm}|p{8.2cm}|p{1.4cm}|}
    \hline
    \multicolumn{1}{|c|}{\textit{Algorithm}}
    & \multicolumn{1}{c|}{\textit{Operations in the base field} \& \textit{Frobenius applications}}
    & \multicolumn{1}{c|}{\textit{Constraints}} \\
    \hline
    \hline

    \textbf{Th.~\ref{theo:motive-charpoly-comp}}
    \tnote{1}
    & $\Otilde(n^2 + (n + r)r^{\Omega - 1})$ \; \& \;
      $O(n^2 + r^2)$
    & \cellcolor{gray!15} \\ \hline

  \end{tabular}

  \begin{tablenotes}
    \footnotesize

    \item [1] Probabilistic algorithm. The characteristic polynomial of any
      endomorphism is the characteristic polynomial of its action on the motive
      of the Drinfeld module.

  \end{tablenotes}

  \end{threeparttable}

\vspace{1cm}

  \begin{threeparttable}

  \captionsetup{font=footnotesize}
  \caption{Algorithms for computing norms of degree $n$ isogenies, in any rank
  $r$, over a finite field of degree $d$ over~$\Fq$\hfill\null}

  \footnotesize
  \begin{tabular}{|p{2.7cm}|p{8.2cm}|p{1.4cm}|}
    \hline
    \multicolumn{1}{|c|}{\textit{Algorithm}}
    & \multicolumn{1}{c|}{\textit{Bit complexity}}
    & \multicolumn{1}{c|}{\textit{Constraints}} \\
    \hline
    \hline

    \textbf{Th.~\ref{theo:iso-norm-finite}}
    \tnote{1}
    & $\Opower\!\big(\big(\SMgeq(n, d) + ndr + n\min(d,r)r^{\omega-1}
               + dr^\omega\big)\log q)
       \precomp$
    & \cellcolor{gray!15} \\ \hline

    \multicolumn{3}{|c|}{See also Table~\ref{table:third}.} \\

    \hline

  \end{tabular}

  \begin{tablenotes}
    \footnotesize

    \item [1] Probabilistic algorithm. The norm of any isogeny is the
      determinant of the motivic application associated to the isogeny.

  \end{tablenotes}

  \end{threeparttable}

\vspace{1cm}

  \begin{threeparttable}

  \captionsetup{font=footnotesize}
  \caption{Algorithms for computing norms of degree $n$ isogenies, in any rank
  $r$ over a finite field of degree $d$ over~$\Fq$\hfill\null}

  \footnotesize
  \begin{tabular}{|p{2.7cm}|p{8.2cm}|p{1.4cm}|}
    \hline
    \multicolumn{1}{|c|}{\textit{Algorithm}}
    & \multicolumn{1}{c|}{\textit{Operations in the base field} \& \textit{Frobenius applications}}
    & \multicolumn{1}{c|}{\textit{Constraints}} \\
    \hline
    \hline

    \textbf{Th.~\ref{theo:iso-norm}}
    \tnote{1}
    & $\Otilde(n^2 + (n + r)r^{\omega-1})$ \; \& \;
      $O(n^2 + r^2)$
    & \cellcolor{gray!15} \\ \hline

  \end{tabular}

  \begin{tablenotes}
    \footnotesize

    \item [1] Probabilistic algorithm. The norm of any isogeny is the
      determinant of the motivic application associated to the isogeny.

  \end{tablenotes}

  \end{threeparttable}
\end{center}

\end{document}